\documentclass{scrartcl}
\usepackage[numbers]{natbib}
\usepackage{fontenc}
\usepackage{shadethm}
\usepackage{amsthm}
\usepackage{float}
\usepackage{bm}
\usepackage{mathtools}
\usepackage{graphicx}
\usepackage{subfig}
\usepackage{stmaryrd}
\usepackage{mathrsfs} 
\usepackage{amsmath}
\usepackage{amssymb}
\usepackage{wasysym}
\usepackage{sectsty}
\usepackage[hyperfootnotes=false]{hyperref}
\usepackage[a4paper, left=2.7cm, right=2.7cm, top=2.5cm]{geometry}
\usepackage{chngcntr}
\counterwithin*{equation}{section}

\usepackage{cleveref}
\DeclareMathAlphabet{\mathpzc}{OT1}{pzc}{m}{it}

\sectionfont{\normalsize\centering}
\subsectionfont{\footnotesize\centering}

\theoremstyle{plain}
\newtheorem{thm}{Theorem}[section] 

\theoremstyle{definition}
\newtheorem{defn}[thm]{Definition} 
\newtheorem{lem}[thm]{Lemma}

\newtheorem{rem}[thm]{Remark}
\newtheorem{cor}[thm]{Corollary}

\def\XXint#1#2#3{{\setbox0=\hbox{$#1{#2#3}{\int}$ }
		\vcenter{\hbox{$#2#3$ }}\kern-.6\wd0}}

\usepackage[utf8]{inputenc}
\usepackage[german,english,russian]{babel}

\newcounter{MPequ}

\newcounter{AppA}

\newcounter{AppB}

\newcounter{AppC}

\newcounter{AppD}

\newcounter{AppE}

\begin{document}\selectlanguage{english}
\begin{center}
\normalsize \textbf{\textsf{Charged particle motion in a strong magnetic field: Applications to plasma confinement}}
\end{center}
\begin{center}
	Ugo Boscain$^{\star,}$\footnote{\textit{E-mail address:} \href{mailto:Ugo.boscain@inria.fr}{ugo.boscain@inria.fr}}, Wadim Gerner$^{\dagger,}$\footnote{\textit{E-mail address:} \href{mailto:wadim.gerner@edu.unige.it}{wadim.gerner@edu.unige.it}}
\end{center}
\begin{center}
{\footnotesize Laboratoire Jacques-Louis Lions, Sorbonne Universit\'{e}, Universit\'{e} de Paris, CNRS, Inria, Paris, France$^{\star}$}
\newline
\newline
{\footnotesize MaLGa Center, Department of Mathematics, Department of Excellence 2023-2027, University of Genoa, Via Dodecaneso 35, 16146 Genova, Italy$^{\dagger}$}
\end{center}
{\small \textbf{Abstract:} 
We derive the zero order approximation of a charged particle under the influence of a strong magnetic field in a mathematically rigorous manner and clarify in which sense this approximation is valid. We use this to further rigorously derive a displacement formula for the pressure of plasma equilibria and compare our findings to results in the physics literature. The main novelty of our results is a qualitative estimate of the confinement time for optimised plasma equilibria with respect to the gyro frequency. These results are of interest in the context of plasma fusion confinement.
\section{Introduction and main results}
\label{I1}
The motion of a charged particle within a magnetic field is governed by the Lorentz force
\begin{gather}
	\nonumber
	F_L=q v\times \mathbf{B}(x)
\end{gather}
where $x=x(t)$ is the particle position, $v=v(t)=\dot{x}(t)$ is the particle velocity, $\mathbf{B}$ is the magnetic field influencing the particle movement and $q$ is the electrical charge of the particle. In addition, Newton's law tells us that $m\ddot{x}=F_L$ and consequently we obtain the equation of motion
\begin{gather}
	\nonumber
	\ddot{x}=\frac{q}{m}\dot{x}\times \mathbf{B}(x).
\end{gather}
To be able to easier compare our findings with the physics literature, we introduce a reference magnetic field strength $|\mathbf{B}_{\operatorname{ref}}|$ which may be regarded as the average magnetic field strength in the region of interest. We then define the gyro-frequency of the particle as $\omega:=\frac{q|\mathbf{B}_{\operatorname{ref}}|}{m}$ and the normalised magnetic field $B:=\frac{\mathbf{B}}{|\mathbf{B}_{\operatorname{ref}}|}$ and arrive at the equation
\begin{gather}
	\label{IE1}
	\ddot{x}_{\omega}=\omega \dot{x}_{\omega}\times B(x_{\omega})
\end{gather}
where the index $\omega$ indicates the dependence of the solution on the parameter $\omega$. In the context of plasma fusion confinement devices one uses a strong magnetic field in order to confine the hot plasma so that the gyro frequency, being proportional to the (average) magnetic field strength, will be very large. It is therefore customary in the physics literature to consider expansions of $x_{\omega}$ for large $\omega\gg 1$, see for instance \cite[Chapter 4]{IGPW24}, \cite[Chapter 3]{Bitt04}. We highlight that the orthogonality of the Lorentz force to the velocity of the particle implies that $|v_{\omega}(t)|^2=|\dot{x}_{\omega}(t)|^2=|v_0|^2$ for all $t\in \mathbb{R}$. This means that the kinetic energy of the particle is preserved in time and thus suggests the possible existence of a limit trajectory.

To be more precise, assume that $B$ is regular enough, fix initial conditions $x_0,v_0\in \mathbb{R}^3$ and let $x_{\omega}$ denote the unique solution to (\ref{IE1}) for given $\omega>0$ (the case of negatively charged particles, i.e. $\omega<0$, can be recovered from the case $\omega>0$ by time reversal). We are then, loosely speaking, interested in finding a Taylor expansion of $x_{\omega}$ around $\frac{1}{\omega}\approx 0$.

Due to the practical relevance for the design of plasma fusion confinement devices, this problem has been thoroughly investigated in the physics literature. However, to the best of our knowledge, a rigorous mathematical treatment of this problem has not yet been conducted.

We focus in the present work, solely on the zero order expansion.  
We recall that the normal particle drift (which makes plasma confinement challenging) appears only in the first order approximation, c.f. \cite[Chapter 4]{IGPW24}. However, the zero order approximation is enough to rigorously derive a formula which imposes constraints on plasma equilibria to be suitable for plasma confinement. 

Before we state our first main result let us consider a simple example. We consider the magnetic field $\mathbf{B}(x):=b(0,0,1)$ where $b\in (0,\infty)$ is a fixed constant. We can then set $|\mathbf{B}_{\operatorname{ref}}|:=b$ and consequently define the gyro frequency by $\omega=\frac{qb}{m}$ where $q$ is the charge and $m$ is the mass of our plasma particle. One can then explicitly solve (\ref{IE1}) and finds the solution
\begin{gather}
	\nonumber
	x_{\omega}(t)=\left(x_{0,1}+\frac{v_{0,2}+v_{0,1}\sin(\omega t)-v_{0,2}\cos(\omega t)}{\omega},x_{0,2}+\frac{-v_{0,1}+v_{0,1}\cos(\omega t)+v_{0,2}\sin(\omega t)}{\omega},x_{0,3}+t v_{0,3}\right)
\end{gather}
where $x_0=(x_{0,1},x_{0,2},x_{0,3})$ and $v_0=(v_{0,1},v_{0,2},v_{0,3})$. This is a helical motion of frequency $\omega$ and radius of order $\frac{1}{\omega}$ around a main motion which is a straight line, \Cref{Spirale}. This justifies the name gyro frequency for $\omega$. Further we observe that $x_{\omega}(t)\rightarrow (x_{0,1},x_{0,2},x_{0,3}+tv_{0,3})$ in $C^{0}(\mathbb{R})$ but that $\dot{x}_{\omega}(t)$ does not converge to $(0,0,v_{0,3})$ in $C^0$.

\begingroup\centering
\begin{figure}[H]
	\hspace{5.5cm}\includegraphics[width=0.2\textwidth]{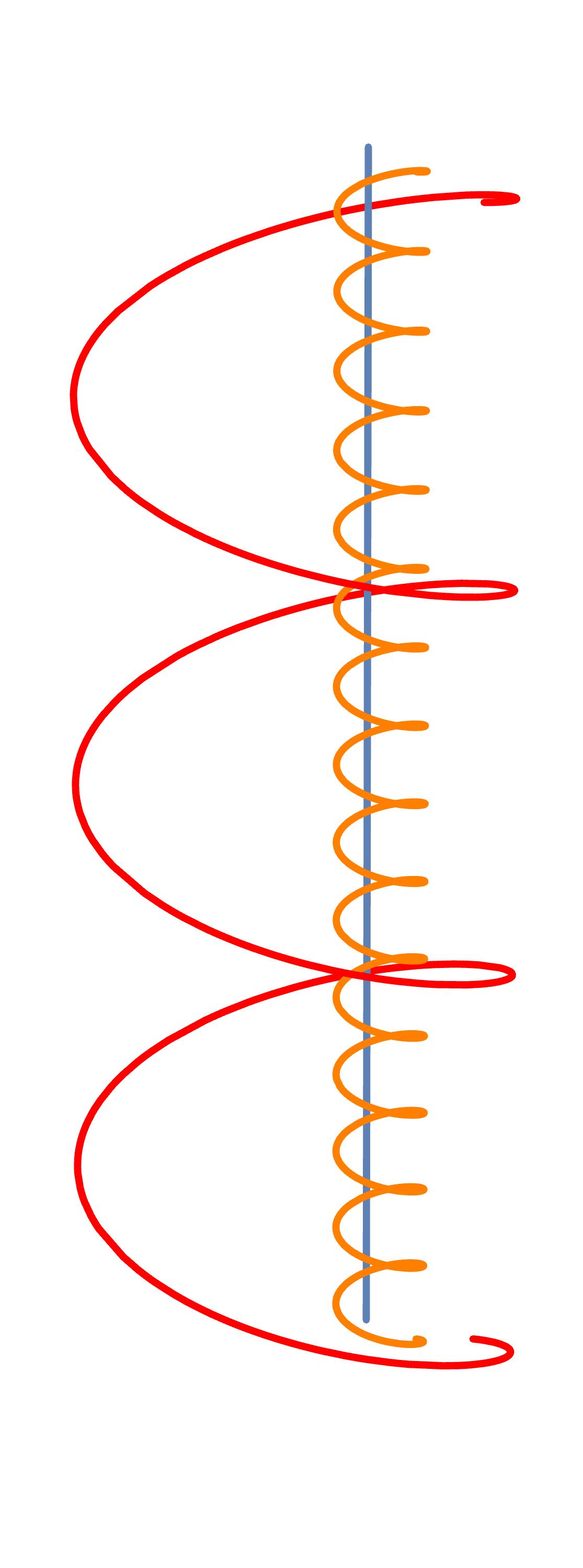}
	\caption{
		Motion of a charged particles in a uniform magnetic field. We plotted the trajectories $x_{\omega}$ for two different values of $\omega$. As $\omega$ increases the radius of the helix decreases and the oscillation becomes faster.} 
	\label{Spirale}
\end{figure}
\endgroup

We will now state our main results and discuss after each statement its connection to known results in the physics literature. Here we say that $B\in C^k$ if it is $k$-times continuously differentiable and the supremums norm of all derivatives up to order $k$ on $\mathbb{R}^3$ is finite. We will write $B\in C^k_{\operatorname{loc}}$ if $B$ is $k$ times continuously differentiable in $\mathbb{R}^3$ (but no assumption on the finiteness of the supremums norm over the full space is made). Lastly, we will need to make assumptions on the decay rate of the magnetic fields in order to be able to derive a qualitative time dependence of the convergence rate.
\begin{defn}[Minimal growth at infinity]
	\label{ID1}
	Let $B\in C^0(\mathbb{R}^3,\mathbb{R}^3)$ be a continuous vector field and $\gamma\in [0,\infty)$. We say that $B$ has a \textit{minimal growth of order $\gamma$ at infinity} if there exist $R,C>0$ such that for all $|y|\geq R$ we have the estimate
	\begin{gather}
		\nonumber
		\frac{C}{|y|^\gamma}\leq |B(y)|.
	\end{gather}
\end{defn} 
\begin{thm}[First Main Theorem: Zero order approximation]
	\label{IT1}
	Let $B\in C^2(\mathbb{R}^3,\mathbb{R}^3)$ be a no-where vanishing, div-free vector field defined on $\mathbb{R}^3$ and fix $x_0,v_0\in \mathbb{R}^3$. Let $x_{\omega}$ be the (unique) solution to (\ref{IE1}) for given $\omega >0$. Then there exists some $x\in C^{3}_{\operatorname{loc}}(\mathbb{R},\mathbb{R}^3)$ such that
	\begin{gather}
		\nonumber
		x_{\omega}\rightarrow x\text{ in }C^{0,\alpha}_{\operatorname{loc}}(\mathbb{R})\text{ for all }0\leq \alpha<1 \text{ as }\omega\rightarrow\infty.
	\end{gather}
	The limit trajectory $x(t)$ satisfies $\dot{x}(t)=h(t)b(x(t))$, where $b(y):=\frac{B(y)}{|B(y)|}$ and $h(t)$ is a suitable $C^2_{\operatorname{loc}}$-function. Further, if $B$ has a minimal growth of degree $\gamma\geq 0$ at infinity, then we have the following qualitative convergence rate
	\begin{gather}
		\nonumber
		\|x_{\omega}-x\|_{C^0[0,t]}\leq C\frac{\exp(C\exp(Ct^{\gamma+2}))}{\omega}\text{ for all }t\geq 0
	\end{gather}
	for a suitable constant $C>0$ which depends on the magnetic field $B$, but is independent of $t$ and $\omega$. Lastly, if in addition $B\in C^k_{\operatorname{loc}}$, then $x\in C^{k+1}_{\operatorname{loc}}$ and if $B$ is analytic, then so is $x$.
\end{thm}
\begin{rem}
	\label{IR}
	\begin{enumerate}
		\item Here we say that $x_{\omega}\rightarrow x$ in $C^{0,\alpha}_{\operatorname{loc}}(\mathbb{R})$ if $x_{\omega}\rightarrow x$ in $C^{0,\alpha}([0,T])$ for any fixed $0<T<\infty$.
		\item Our result regarding the type of convergence is optimal in the sense that if $x_{\omega}\rightarrow x$ in $W^{1,1}_{\operatorname{loc}}(\mathbb{R},\mathbb{R}^3)$, then in particular $v_{\omega}=\dot{x}_{\omega}\rightarrow \dot{x}$ in $L^1_{\operatorname{loc}}(\mathbb{R})$ and consequently (after possibly passing to a subsequence) $v_{\omega}(t)\rightarrow \dot{x}(t)$ for a.e. $t$. It follows however from (\ref{IE1}) that $|v_{\omega}(t)|=|v_0|$ for all $t$ and $\omega$ so that we infer $|\dot{x}(t)|=|v_0|$ for almost all $t$ and by continuity of $\dot{x}(t)$ we find $|\dot{x}(t)|=|v_0|$ for all $t\geq 0$ and in particular $|\dot{x}(0)|=|v_0|$. But we will see in the proof of \Cref{IT1} that we have $|\dot{x}(0)|=|v_0\cdot b(x_0)|$ and so as soon as $v_0$ is not parallel to $b(x_0)$ we cannot have a $W^{1,1}_{\operatorname{loc}}$ convergence. As a consequence, $x_{\omega}(t)$ does not converge to $x(t)$ in $C^{0,1}_{\operatorname{loc}}(\mathbb{R})$ either.
		\item We further note that in contrast to the situation of a uniform magnetic field $B=(0,0,1)$ one cannot improve the $C^{0}_{\operatorname{loc}}(\mathbb{R})$ convergence to a $C^0(\mathbb{R})$ convergence in general. Indeed, one may consider a vector field $B(x,y,z):=(-y,x,f(x^2+y^2))$ where $f$ is a smooth bump function which is non-zero at $0$ and compactly supported within $(-1,1)$. This defines a no-where vanishing div-free field and one can show that for initial conditions $x_0,v_0$ with $x^2_{0,1}+x^2_{0,2}>1$ (where $x_0=(x_{0,1},x_{0,2},x_{0,3})$), $x_{\omega}(t)$ will approach a limit trajectory $x(t)$ which is a reparametrisation of a field line of $B$. More precisely, the function $h$ from \Cref{IT1} will be constant and given by $h(t)=v_0\cdot b(x_0)$ in this specific situation. One can further show that $|x^3_{\omega}(t)-x_{0,3}|\in \Theta\left(\frac{\dot{\phi}^2_0t}{\omega}\right)$, i.e. both terms can be bounded by each other modulo some time and $\omega$-independent constant, where $x_{\omega}(t)=(x^1_{\omega}(t),x^2_{\omega}(t),x^3_{\omega}(t))$ and where $\dot{\phi}_0$ denotes the initial angular velocity. Since field lines of $B$ in this situation do not change the $z$-component of the initial condition, we conclude that as soon as $\dot{\phi}_0\neq 0$ we will not have a $C^0(\mathbb{R})$ convergence.
		\item The function $h(t)$ in \Cref{IT1} turns out to be constant for any initial condition $x_0,v_0\in \mathbb{R}^3$ if and only if $\operatorname{div}(b)(y)=0$ for all $y\in \mathbb{R}^3$.
		\item In the context of plasma fusion confinement it is reasonable to assume that the (volume) current density which generates the magnetic field within the plasma confinement device is compactly supported within some compact region of $\mathbb{R}^3$ (the massive $3$-d coils). The induced magnetic field is given by the Biot-Savart law and it is well-known that in this situation the magnetic field $B$ and its derivatives decay to zero at infinity so that the globally $C^2$- and div-free assumptions on the magnetic field in the first part of \Cref{IT1} are natural.
		\item The additional non-vanishing and minimal growth at infinity assumptions are of technical nature. The necessity for these assumptions stems from the fact that the error term will contain terms of the form $\frac{1}{|B(x_{\omega}(t))|}$. For general magnetic fields (which may support non-compact field lines), due to the convergence $x_{\omega}(t)\rightarrow x(t)$ in $C^0_{\operatorname{loc}}$, $x_{\omega}$ will in general not be contained in any finite subregion of $3$-space and hence, in order to estimate terms like $\frac{1}{|B(x_{\omega}(t))|}$ it will be essential to make assumptions on the decay of $|B|$ at infinity.
		\item If the magnetic field $B$ admits zeros, then a generalisation of \Cref{IT1} is possible as long as we restrict attention to initial conditions $x_0,v_0$ for which the trajectories $x_{\omega}(t)$ avoid the zero set of $B$ for all times $t\geq 0$ and all large enough $\omega \gg 1$.
		\item One can further consider the situation in which $B$ is not defined on all of $\mathbb{R}^3$ but only on a subset $\Omega\subset\mathbb{R}^3$. However, in this case it is no longer guaranteed that the trajectories $x_{\omega}$ are defined for all times and consequently one can only consider times $T>0$ for which the trajectories $x_{\omega}$ are all defined up to the same time $T$ for all large enough $\omega$.
		\item In the context of compactly supported current distributions, it has been shown in \cite[Section III A]{CDG01}, that the magnetic field decays at least as fast as $\frac{1}{|y|^3}$ at infinity, so that the best possible estimate we can hope for by means of \Cref{IT1} is obtained by picking $\gamma=3$ for the applications which we have in mind.
		\item We do not claim any optimality regarding the time dependence of the convergence rate. The reason we arrive at a "double-exponential" estimate arises informally speaking from the fact that we have to deal with a coupled system of 2 first order ODEs. In order to gain qualitative control we will apply Gronwall's inequality twice which leads us to the appearance of two exponential functions. In fact, our method allows us to infer a slightly better qualitative bound, which is however a bit more complicating to write down and is still far from being "single-exponential" so that we decided to use the double exponential form instead.
	\end{enumerate}
\end{rem}
As we will see in the proof of \Cref{IT1} we also obtain the following property of the limit trajectory $x$.
\begin{cor}[Magnetic moment preservation]
	\label{IC2}
	Let $B\in C^2(\mathbb{R}^3,\mathbb{R}^3)$ be a no-where vanishing div-free vector field and let $x(t)$ and $h(t)$ be as in the statement of \Cref{IT1} for fixed $x_0,v_0\in \mathbb{R}^3$. Define $\mu(t):=\frac{|v_0|^2-h^2(t)}{2|B(x(t))|}$, then
	\begin{gather}
		\nonumber
		\mu(t)=\mu(0)\equiv \mu_0 \text{ for all } t\geq 0.
	\end{gather}
\end{cor}
In the context of physics literature the zero order approximation $x$ of $x_{\omega}$ is known as the zeroth order order guiding center motion, see for instance \cite[Introduction]{BMKN23}, which is known to follow magnetic field lines. We note that the notion of trapped particles in plasma physics corresponds to the situation where the function $h$ in \Cref{IC2} changes signs and consequently the zero order approximation is "reflected".

As for the physical interpretation of \Cref{IC2} we note that the speed $|v_\omega|$ is preserved in time for all solutions of (\ref{IE1}), i.e. $|v_{\omega}(t)|=|v_0|$ for all times $t$ and any fixed $\omega$. So if we let $v^\parallel_{\omega}(t):=(v_{\omega}(t)\cdot b(x_{\omega}(t)))b(x_{\omega}(t))$ and $v^\perp_{\omega}(t):=v_{\omega}(t)-v^\parallel_{\omega}(t)$, we can define the magnetic moment of the particle by $\mu_{\omega}(t):=\frac{|v^\perp_{\omega}(t)|^2}{2|B(x_{\omega}(t))|}$, see \cite[Chapter 4]{IGPW24}. Since $|v^\perp_{\omega}|^2=|v_{\omega}|^2-(v^\parallel_{\omega})^2=|v_0|^2-h^2_{\omega}$, where we set $h_{\omega}(t):=v_{\omega}(t)\cdot b(x_{\omega}(t))$, it will follow from the proof of \Cref{IT1} not only that $x_{\omega}\rightarrow x$ in $C^{0,\alpha}_{\operatorname{loc}}(\mathbb{R})$, but also that $h_{\omega}(t)\rightarrow h(t)$ in $C^{0,\alpha}_{\operatorname{loc}}(\mathbb{R})$ so that
\begin{gather}
	\label{IE2}
	\lim_{\omega\rightarrow\infty}\mu_{\omega}(t)=\mu(t)\text{ in }C^{0,\alpha}_{\operatorname{loc}}(\mathbb{R})\text{ for all }0\leq\alpha<1.
\end{gather}
\Cref{IE2} can therefore be interpreted in the sense that the magnetic moment $\mu_{\omega}(t)$ of a charged particle is an almost preserved quantity in the context of high gyro frequencies and thus justifies rigorously the fact that the magnetic moment is referred to as an adiabatic invariant in physics, \cite[Chapter 4.2]{IGPW24}.

Our second main result concerns the change of pressure along a charged particle trajectory. To be more precise we recall here that for a given bounded $C^1$-domain $\Omega\subset\mathbb{R}^3$ we say that a $C^1$-vector field $B$ on $\Omega$ is a plasma equilibrium if $B$ satisfies
\begin{gather}
	\label{IE3}
	\operatorname{div}(B)=0\text{, }B\times \operatorname{curl}(B)=\nabla p\text{, }B, \operatorname{curl}(B) \parallel \partial\Omega
\end{gather}
where $p$ is the "pressure"-function of the system and $B$ is the magnetic field. We point out that it follows immediately from (\ref{IE3}) that $B$ as well as $\operatorname{curl}(B)$ are tangent to the level sets of $p$. Usually the condition that $\operatorname{curl}(B)$ is parallel to the boundary of $\Omega$ is not included in the definition of a plasma equilibrium. The reason we impose this condition is to guarantee the absence of invariant cylindrical surfaces. To be more precise, according to Arnold's structure theorem \cite[Section II Theorem 1.2]{AK21}, assuming $p$ is not constant and all quantities being analytic, there exists an at most $2$-dimensional subset $S\subset \Omega$ such that $\Omega\setminus S$ consists of finitely many connected components, each of which is foliated into surfaces which are invariant under the flow of $B$ and $\operatorname{curl}(B)$. Within any such component, either all leafs of the foliation are tori and the fields $B$ and $\operatorname{curl}(B)$ linearise on any such surface, or all the leafs are cylindrical surfaces, in which case the field lines of $B$ are all closed on any such cylinder. Further, the invariant surfaces correspond to regular level sets of the pressure $p$. The additional assumption $\operatorname{curl}(B)\parallel \partial \Omega$ rules out the existence of components consisting of cylindrical leafs.

Now, to come to our next main result, we want to understand how far a particle $x_{\omega}$ may drift away from an initial pressure surface $p_0=p(x_0)$ under the influence of a plasma equilibrium field, i.e. we would like to find an expansion for $p(x_{\omega}(t))$ in terms of $\frac{1}{\omega}$. We summarise our main findings in the following theorem.
\begin{thm}[Second Main Theorem: Displacement formula]
	\label{IT3}
	Let $B\in C^3$ be a no-where vanishing, div-free field on $\mathbb{R}^3$ and suppose that there is a bounded domain $\Omega\subset\mathbb{R}^3$, $\partial\Omega\in C^3$ and a function $p\in C^{3}(\overline{\Omega})$ such that $B,p$ satisfy (\ref{IE3}). Then for any given $x_0\in \Omega,v_0\in \mathbb{R}^3$ we have the expansion
	\begin{gather}
		\nonumber
		p(x_{\omega}(t))=p(x_0)+\frac{(b(x(t))\times v_{\omega}(t))\cdot \nabla p (x(t))}{|B(x(t))|\omega}-\frac{(b(x_0)\times v_0)\cdot \nabla p(x_0)}{\omega|B(x_0)|}
		\\
		\nonumber
		+\frac{1}{\omega}\int_0^t\left(\frac{|v_0|^2}{|B(x(s))|}-\mu_0\right)(B(x(s))\times \nabla p(x(s)))\cdot \nabla (|B|^{-1})(x(s))ds+\frac{\mathcal{O}(1)}{\omega^2}\exp(C\exp(Ct^2))
	\end{gather}
	for a suitable $C>0$ which is independent of $t$ and $\omega$. Here $x_{\omega}(t)$ is the solution to (\ref{IE1}) with initial conditions $x_0,v_0$; $x(t)$ is the limit trajectory from \Cref{IT1}, $\mu_0$ is the magnetic moment as defined in \Cref{IC2}, the term $\mathcal{O}(1)$ is uniformly bounded in $\omega$ and $t$, i.e. there is some $c>0$ depending on $B$ and $p$ but independent of $t$ and $\omega$ such that $|\mathcal{O}(1)|\leq c$ and where this expansion is valid for all times $t$ and $\omega$ for which $x_{\omega}(s)\in \Omega$ for all $0\leq s\leq t$.
\end{thm}
Regarding the last condition, we note that for any $x_0\in \Omega$ and since $B\parallel \partial \Omega$ and the limit trajectory $x$ moves parallel to the $B$ field, it will be the case that there exists some $\tau(x_0)>0$ such that $x_{\omega}(t)\in \Omega$ for all $0\leq t\leq \tau$ and all $\omega \gg 1$. The condition that $B$ should be defined on all of $\mathbb{R}^3$ is of technical nature to be able to make sense of the trajectories of $x_{\omega}$ for all times and be able to apply \Cref{IT1}.

The main novelty of \Cref{IT3} with regards to what is already known in the physics literature is the explicit control of the error term in terms of the time. Usually, in the approaches taken in the physics literature, higher order terms are simply discarded and the corresponding error not analysed, so that it remains unclear up to what time scale a first order approximation remains a good approximation. Here we can see that the leading order correction term of order $\frac{1}{\omega}$ will be dominant up to a timescale of order $\sqrt{\ln(\ln(\omega))}$. Consequently, up to this time scale it is reasonable to focus solely on the leading order correction term. Again, we do not claim optimality of the derived time-scale. We will see in the proof of \Cref{IT3} that the error term contains some terms which grow only linearly in time as well as more complex terms which we can only control "double-exponentially". This seems to suggest that the optimal convergence rate in the context of \Cref{IT3} grows at best linearly in time.

Let us point out that it is more common in the physics literature to consider the toroidal flux function $\psi$ rather than the pressure $p$, but due to the structure of the plasma equilibrium equations we believe it is more natural to formulate the results in terms of the pressure function.
\begin{rem}
	A few comments regarding \Cref{IT3} are due
	\begin{enumerate}
		\item We require a higher regularity of the $B$ field in \Cref{IT3} in comparison to \Cref{IT1} because in \Cref{IT3} we essentially need to compute the 3rd time derivative of $p(x_{\omega}(t))$ which stems from integration by parts. While we also require to take 3rd order derivatives in \Cref{IT1}, in this case we essentially have to differentiate $v_{\omega}(t)$ to gain control on its parallel component. However, due to (\ref{IE1}) we see that $\dot{v}_{\omega}$ involves only $B$ and none of its derivatives so that we end up requiring only twice differentiability of the magnetic field.
		\item The qualitative time-dependent error estimate in \Cref{IT3} corresponds to the best possible estimate with $\gamma=0$ in \Cref{IT1} while we at the same time do not impose any growth rate condition at infinity. The reason for this is that \Cref{IT3} is only valid for times $t$ up to which $x_{\omega}(t)$ remains within $\Omega$ which is a bounded region. Hence, for all such times we can control terms of the form $|B(x_{\omega}(t))|^{-1}$ by constants which depend only on $B$ and $\Omega$ so that we may assume that the magnetic field in question has a minimal growth of degree zero at infinity.
		\item We do not require the quantities to be analytic or the pressure to be non-constant. Even if $p(x_0)$ is a non-regular pressure level set our formula applies.
		\item We notice that one sufficient condition for the leading order correction term to vanish is to demand that $(B(y)\times \nabla p(y))\cdot \nabla |B|^{-1}(y)=0$ for all $y\in \Omega$. This is easily seen on each regular level set of $p$ to be equivalent to demand that $|B|$ is a first integral of $\mathcal{N}\times B$ where $\mathcal{N}$ is a unit normal along the regular level set (which is a torus). This leads to the notion of isodynamic stellarators which has been well-studied in the physics literature, c.f. \cite{Pal68},\cite{Pal86}, and turned out to be a too restrictive condition to be utilisable in real life stellarator designs, c.f. \cite{Cha86},\cite{BerMosTat86}. We also point out that the fact that $\operatorname{curl}(B)$ is also tangent to the level sets of $p$ implies that $\mathcal{N}\times B$ is div-free as a vector field on the level set and consequently, \cite[Theorem 9]{PPS22}, since $B$ does not vanish on any regular level set, $\mathcal{N}\times B$ is semi-rectifiable so that either all of its field lines are closed or all of its field lines are dense on any fixed regular level set. Further, to put it in physical terms, unless the rotational transform of $\mathcal{N}\times B$ remains constant (and rational) along all regular level sets, it will turn out that $|B|$ is constant on all irrational $\mathcal{N}\times B$ surfaces and by density on all regular level sets, which is also sometimes taken as a definition of isodynamicity \cite[\S 1b]{Pal86}. We in particular see that for an isodynamic stellarator we have $|p(x_{\omega}(t))-p(x_0)| \lesssim \frac{1}{\omega}$ for a time frame of order $t\sim \sqrt{\ln(\ln(\omega))}$, i.e. the level set changes only by a magnitude of order $\frac{1}{\omega}$ for a time scale of order $\sqrt{\ln(\ln(\omega))}$ and hence remains confined for all this time.
		\item We can make the last comment more precise. Let us say that a plasma equilibrium $B$ is (ideally) optimised if for every $\nu>0$ there is a constant $c_{\nu}>0$ such that
		\begin{gather}
	\nonumber
	\left|\int_0^t\left(\frac{|v_0|^2}{|B(x(s))|}-\mu_0\right)(B(x(s))\times \nabla p (x(s)))\cdot \nabla \left(|B|^{-1}\right)(x(s))ds\right|\leq c_{\nu}\text{ for all }t\geq 0
	\end{gather}
		and all $x_0\in \Omega$, $|v_0|\leq \nu$. Then, given any large enough $\omega>0$ there will exist some $T>0$ such that (for fixed $x_0$ and $v_0$) $x_{\omega}([0,T])\subset \Omega$ which follows from \Cref{IT1} and the fact that field lines of $B$ are confined within compact subsets of $\Omega$. We can then define the confinement time by $\tau_{\omega}:=\inf\{t\geq 0\mid x_{\omega}(t)\notin \Omega\}$. If $B$ is an optimised plasma equilibrium, the first order correction from \Cref{IT3} will be of order $\frac{\mathcal{O}(1)}{\omega}$ uniformly in $t$ and the remaining error term will be of the same order as long as $t$ is of order $\sqrt{\ln(\ln(\omega))}$. We recall that we assume that $\operatorname{curl}(B)\parallel \partial\Omega$ which implies that $p$ is constant along $\partial\Omega$, $p(y)=p_b$ for all $y\in \partial\Omega$. If we assume in addition that $\partial\Omega=\{p=p_b\}$, then we see that for large enough $\omega$ it will be the case that $|p_b-p_0|>0$ as long as $t$ is of order $\sqrt{\ln(\ln(\omega))}$ and thus we find $\tau_{\omega}\gtrsim \mathcal{\sqrt{\ln(\ln(\omega))}}$ providing us with a qualitative lower bound of the confinement time.
		\item The term $\frac{(b(x(t))\times v_{\omega}(t))\cdot \nabla p(x(t))}{|B(x(t))|}$ still contains an $\omega$ dependence in $v_{\omega}$ and should be interpreted as an oscillating motion similar in spirit to the gyro motion in the guiding centre particle description in the physics literature \cite[Chapter 4]{IGPW24}. In fact, as soon as we average the displacement, i.e. consider $\frac{1}{T}\int_0^Tp(x_{\omega}(t))dt$ for any fixed $T>0$ it turns out that the corresponding integral containing the term $b(x(t))\times v_{\omega}(t)$ will be of order $\frac{1}{\omega^2}$ and hence not contribute significantly to the average pressure change along a particle trajectory.
	\end{enumerate}
\end{rem}

Our final result aims to draw attention to the existence of "resonant surfaces" on, and near which, even optimised plasma equilibria show poor confinement properties.

We observe first that according to \Cref{IT3} the leading order correction term is at worst of order $\frac{t}{\omega}$ modulo a constant which depends on $B$ and $p$ but is independent of $t$ and $\omega$, so that even non-optimised plasma equilibria drift away at worst linearly in time from the initial pressure surface.
\newline
\newline
It is customary in the physics literature to introduce the notion of omnigenity, which demands that for any periodic zero order trajectory approximation, the average drift of $x_{\omega}$ perpendicular to a given regular surface vanishes along a period, c.f. \cite{LanCat12},\cite{PCHL15}.

It is also well-known that passing particles on irrational pressure surfaces, i.e. surfaces on which all magnetic field lines are dense and particles whose initial conditions $x_0,v_0$ satisfy $|v_0|^2-2|B(y)|\mu_0>0$ for all $y$ on the regular level set $p(x_0)$, where $\mu_0=\frac{|v_0|^2-(b(x_0)\cdot v_0)^2}{2|B(x_0)|}$ and $b(y)=\frac{B(y)}{|B(y)|}$, satisfy the property that the average perpendicular drift vanishes in the limit $t\rightarrow\infty$ when the average is taken up to time $t$, \cite[Appendix C.3]{BMKN23}.

We want to point out that the average drift being zero in the limit $t\rightarrow\infty$ is not a sufficient condition for the average change of a function to stay bounded in the limit $t\rightarrow\infty$, i.e. even though the average drift vanishes in the limit, it is of importance to understand the rate at which the limit is approached, since otherwise the particle might still escape to infinity as $t$ becomes larger.

We shall shortly see that even in the case of quasi-symmetric plasma equilibria, there exists a level set on which the drift away from said level set grows linearly in time, which, as pointed out previously, is the worst possible drift in time. Near this "resonant" surface one can still find a uniform upper bound on the first order correction term independent of $t$, but it gets worse and worse as we approach the resonant surface.

This suggests that the situation for general omnigenious plasma equilibria might be even worse, in the sense that there might exist even more resonant surfaces. These observations suggest a negative answer to a conjecture in \cite[Appendix C.3]{BMKN23} that for omnigenious plasma equilibria a uniform convergence rate of the drift rate on irrational surfaces might be achievable, and hence further efforts must be made in order to ensure that passing particles on irrational surfaces stay confined.
\newline
\newline
Before we state our last result, we recall here shortly some notions pertaining the definition of quasi-symmetric plasma equilibria. We assume throughout that $B$ is a no-where vanishing, div-free, $C^3$-vector field on $\mathbb{R}^3$ which satisfies (\ref{IE3}) in some bounded $C^3$-region $\Omega$ and suitable $C^3$-function $p$. We suppose that $x_0\in \Omega$ is chosen such that $p_0:=p(x_0)$ is a regular level set. This level set $\{p=p_0\}$ then admits an open neighbourhood $U$ consisting only of regular level sets.

On each such neighbourhood $U$ we can define the following vector fields
\begin{gather}
	\label{IE4}
	\eta:= \frac{B}{|B|^2}\text{ and }\xi:=\eta \times \nabla p.
\end{gather}
It is then standard to verify that $[\eta,\xi]=0$ where $[\cdot,\cdot]$ denotes the Lie-bracket of vector fields. Further, by construction, $\xi$ is tangent to the level sets of $p$ and similarly, since $B$ is tangent to the level sets of $p$, $\eta$ is tangent to the level sets of $p$. Now by the construction of action angle coordinates, also known as Arnold-Liouville coordinates, c.f. \cite[Chapter 10]{A89}, we can find a global coordinate chart $\Psi: I\times T^2\rightarrow \Omega$ (where $I$ is some open interval) in which the vector fields $\eta$ and $\xi$ take the form
\begin{gather}
	\label{IE5}
	\eta= \alpha(p) \partial_{\phi}+\beta(p)\partial_{\theta}\text{, }\xi=a(p)\partial_{\phi}+c(p)\partial_{\theta}
\end{gather}
where $\partial_{\phi},\partial_{\theta}$ are the standard coordinate fields on $T^2$ chosen such that the field lines of $\partial_{\phi}$ are homotopic to a poloidal closed curve on the level sets of $p$ and $\partial_{\theta}$ are chosen such that its field lines correspond to a toroidal field line and where $\alpha,\beta,a,c$ are suitable functions which depend solely on the pressure label, i.e. they are constant on every fixed regular level set. To clarify, we say here that a closed curve on a toroidal surface $T^2\cong \Sigma\hookrightarrow \mathbb{R}^3$ is poloidal if it is not contractible within $\Sigma$ and bounds a surface within the finite region bounded by $\Sigma$. Similarly we say that a closed curve on $\Sigma$ is toroidal if it is non-trivial and bounds a surface within the unbounded region bounded by $\Sigma$. The above mentioned action angle coordinates which "linearise" $\eta$ and $\xi$ are referred to as Boozer coordinates, c.f. \cite[Section 2.5]{Hel14}. See also \cite[Part II]{DHCS91} for a detailed discussion of different types of distinguished coordinate systems in plasma fusion.

Quasi-symmetric equilibria can then be defined as follows.

\begin{defn}[Quasi-symmetric plasma equilibrium]
	\label{ID5}
	Let $B\in C^{3}(\mathbb{R}^3,\mathbb{R}^3)$ be a no-where vanishing, div-free vector field which satisfies (\ref{IE3}) within some bounded $C^3$-region $\Omega$ and some $p\in C^{3}(\overline{\Omega},\mathbb{R})$. Then $B$ is called quasi-symmetric, if around any given regular level set of $p$ and fixed Boozer coordinate system, there exist integers $(M,N)\in \mathbb{Z}^2\setminus (0,0)$ and a function $f:I\times S^1\rightarrow [0,\infty)$ such that we have $|B(p,\phi,\theta)|=f(p, M\phi+N\theta)$ for the expression of $|B|$ in the chosen Boozer coordinate system. 
\end{defn}
We note that Boozer coordinates are not unique, because we can always apply the transformations $\phi\rightarrow -\phi$ and $\theta\rightarrow -\theta$, but if we specify the orientation of the poloidal and toroidal curves, the choice becomes unique if we choose $I$ to be the maximal interval of the well-definedness of the action-angle coordinates (up to translations $\phi\rightarrow \phi+\phi_0$, $\theta\rightarrow\theta+\theta_0$ for fixed $\phi_0,\theta_0\in \mathbb{R}$).

Quasi-symmetric plasma equilibria are of great importance for plasma confinement, c.f. the pioneering works \cite{BoozPytt81},\cite{Booz83},\cite{Booz95} and are still an active area of research today, e.g. \cite{LanPaul22},\cite{Landre22}.
\newline
\newline
Our last result is the following, where we recall the following conventions:
\begin{enumerate}
	\item For given $x_0,v_0\in \mathbb{R}^3$, $x_{\omega}$ denotes the solution of (\ref{IE1}) for fixed $x_0,v_0$ and $\omega>0$.
	\item $x(t)$ denotes the limit trajectory as in \Cref{IT1} and $h(t)$ corresponds to the speed of $x$, recall \Cref{IC2}. 
	\item As usual we also let $b(y):=\frac{B(y)}{|B(y)|}$.
	\item We have $\mu_0=\frac{|v_0|^2-(v_0\cdot b(x_0))^2}{2|B(x_0)|}$.
	\item We say that a pair of initial conditions $(x_0,v_0)$ are passing for a plasma equilibrium, if $|v_0|^2-2|B(y)|\mu_0>0$ for all $y\in \{p(z)=p(x_0)\}$.
\end{enumerate}
\begin{thm}[Third Main Theorem: Resonant surfaces]
	\label{IT6}
	Let $B\in C^3(\mathbb{R}^3,\mathbb{R}^3)$ be a no-where vanishing div-free vector field satisfying (\ref{IE3}) for a bounded region $\Omega$, $\partial\Omega\in C^3$, and a function $p\in C^{3}(\overline{\Omega})$. Further, let $x_0\in \Omega$ be such that $p_0:=p(x_0)$ is a regular level set and $v_0\in \mathbb{R}^3$ such that $(x_0,v_0)$ are passing. If $B$ is quasi-symmetric, then for a suitable $C>0$ which is independent of $t$ and $\omega$ we have
	\begin{gather}
		\nonumber
		p(x_{\omega}(t))=p_0+\frac{(b(x(t))\times v_{\omega}(t))\cdot \nabla p(x(t))}{|B(x(t))|\omega}-\frac{(b(x_0)\times v_0)\cdot \nabla p(x_0)}{|B(x_0)|}
		\\
		\nonumber
		+\operatorname{sgn}(v_0\cdot b(x_0))\frac{a(p_0)M+c(p_0)N}{\omega (\alpha(p_0) M+\beta(p_0) N)}\left(\frac{|h(t)|}{|B(x(t))|}-\frac{|v_0\cdot b(x_0)|}{|B(x_0)|}\right)+\frac{\mathcal{O}(1)}{\omega^2}\exp(C\exp(Ct^2))
	\end{gather}
	whenever $\alpha(p_0) M+\beta(p_0) N\neq 0$ and
	\begin{gather}
		\nonumber
		p(x_{\omega}(t))=p_0+\frac{(b(x(t))\times v_{\omega}(t))\cdot \nabla p(x(t))}{|B(x(t))|\omega}-\frac{(b(x_0)\times v_0)\cdot \nabla p(x_0)}{|B(x_0)|}
		\\
		\nonumber
		-\operatorname{sgn}(v_0\cdot b(x_0))\frac{(a(p_0)M+c(p_0)N)(|v_0|^2-\mu_0|B(x_0)|)\rho(t)}{\omega |B(x_0)|^2\sqrt{|v_0|^2-2\mu_0|B(x_0)|}}f^{\prime}(M\phi_0+N\theta_0)+\frac{\mathcal{O}(1)}{\omega^2}\exp(C\exp(Ct^2))
	\end{gather}
	whenever $\alpha(p_0) M+\beta(p_0) N=0$, where $\rho(t)=\int_0^t|B(x(s))|\sqrt{|v_0|^2-2\mu_0|B(x(s))|}ds$, where $\alpha,\beta,a,c$ are the functions in the expression (\ref{IE5}), $\operatorname{sgn}(s)=+1$ if $s>0$ and $\operatorname{sgn}(s)=-1$ if $s<0$ and where $|B|=f$ as in \Cref{ID5}, $(p_0,\phi_0,\theta_0)$ are the Boozer coordinate expression of $x_0$ and the derivative of $f$ is computed w.r.t. to its second variable.
\end{thm}
It is customary in physics literature to assume that the plasma equilibria in question have a non-vanishing toroidal component which amounts to assuming that $\beta\neq 0$ and one then defines the rotational transform $\iota$ of the field lines of $B$ by $\iota(p):=\frac{\alpha(p)}{\beta(p)}$. \Cref{IT6} therefore tells us that, since $|h(t)|\leq |v_0|$, as long as the rotational transform of $B$ is bounded away from $-\frac{N}{M}$ we find that $p(x_{\omega}(t))\approx p_0+\mathcal{O}\left(\frac{1}{\omega}\right)$ for a time frame of order $t\sim \sqrt{\ln(\ln(\omega))}$. On the other hand, if we are on a "resonant" surface satisfying $\iota(p_0)=-\frac{N}{M}$, we see that $\rho(t)\sim t$ because we are dealing with a passing particle and further $aM+cN\neq 0$ on this surface because $\eta$ and $\xi$, c.f. (\ref{IE4}),(\ref{IE5}), are linearly independent. This shows that, unless $|B|$ is constant on the resonant surface, there will be initial conditions such that $p(x_{\omega}(t))\approx p_0+\frac{t\Theta(1)}{\omega}$ so that the particle will drift away linearly in time from the resonant surface which, as we have pointed out before, is the worst possible drift behaviour.

In conclusion, even for highly optimised plasma equilibria, such as quasi-symmetric equilibria one should make either sure that the hot plasma avoids coming close to these resonant surfaces or otherwise ensure that the magnetic field strength is constant on these surfaces to obtain good confinement across all pressure surfaces.
\section{Proof of \Cref{IT1} \& \Cref{IC2}}
We start by establishing the well-posedness of the initial value problem (\ref{IE1}).
\begin{lem}[Existence of global particle path]
	\label{L1}
	Let $B\in C^{0,1}_{\operatorname{loc}}(\mathbb{R}^3,\mathbb{R}^3)$. Then for every $x_0,v_0\in \mathbb{R}^3$ and every $\omega\in \mathbb{R}$ the following initial value problem
	\begin{gather}
		\label{E1}
		\ddot{x}(t)=\omega \dot{x}(t)\times B(x(t))\text{, }x(0)=x_0\text{ and }\dot{x}(0)=v_0
	\end{gather}
	admits a unique solution of class $C^{2,1}_{\operatorname{loc}}(\mathbb{R})$ defined for all times $t\in \mathbb{R}$.
\end{lem}
\begin{proof}[Proof of \Cref{L1}]
	We consider the vector field $V:\mathbb{R}^3\times \mathbb{R}^3\rightarrow \mathbb{R}^3\times\mathbb{R}^3$, $(x,v)\mapsto (v,\omega v\times B(x))$ and observe that since $B\in C^{0,1}_{\operatorname{loc}}$ we find $V\in C^{0,1}_{\operatorname{loc}}(\mathbb{R}^3\times \mathbb{R}^3,\mathbb{R}^3\times \mathbb{R}^3)$ and thus it admits a maximal $C^1$-integral curve starting at $(x_0,v_0)$. Then $\dot{v}(t)=\omega v(t)\times B(x(t))$ is in turn of class $C^{0,1}$ and consequently $v$ of class $C^{1,1}$ and so $x$ of class $C^{2,1}$. To see that this flow is global we observe that $|v(t)|^2$ is preserved in time, i.e. $|v(t)|=|v_0|$ for all $t$ for which the integral curve is defined and further, $|x(t)|=\left|x_0+\int_0^tv(s)ds\right|\leq |x_0|+\int_0^t|v(s)|ds=|x_0|+|v_0||t|$ and therefore $(x(t),v(t))$ is always contained in a compact subset of $\mathbb{R}^3\times \mathbb{R}^3$ for all finite times $t$. Thus the escape lemma, \cite[Lemma 9.19]{L12}, implies that the integral curve is global.
\end{proof}
\begin{proof}[Proof of \Cref{IT1}]
	The proof consists of 3 steps. In the first step we fix any sequence $\omega_n\rightarrow\infty$ and show that for any fixed $0\leq \alpha<1$ and any given subsequence of $(\omega_n)_n$ there exists yet again another subsequence converging in $C^{0,\alpha}_{\operatorname{loc}}(\mathbb{R})$ to some limit function $x$ which, a priori, might be distinct for distinct initial subsequences.
	
	In the second step we show that this limit function satisfies a second order ODE which by uniqueness of solutions to ODEs will show that this limit is in fact independent of the initially chosen subsequence so that the convergence of the original sequence will follow from the subsequence principle.
	
	In the last step we will establish the claimed convergence rate.
	\newline
	\newline
	\underline{Step 1:} We start by setting $h_{\omega}(t):=v_{\omega}(t)\cdot b(x_{\omega}(t))$, $v^\parallel_{\omega}(t):=h_{\omega}(t)b(x_{\omega}(t))$, $v^\perp_{\omega}(t):=v_{\omega}(t)-v^\parallel_{\omega}(t)$ for fixed $\omega>0$, where $b(y):=\frac{B(y)}{|B(y)|}$. We can therefore write
	\begin{gather}
		\nonumber
		x_{\omega}(t)-x_0=\int_0^t v_{\omega}(s)ds=\int_0^tv^\parallel_{\omega}(s)ds+\int_0^tv^\perp_{\omega}(s)ds.
	\end{gather}
	Our goal now is to show that the $C^{0,1}_{\operatorname{loc}}(\mathbb{R})$-norm of $x_{\omega}$ is bounded in the sense that its $C^{0,1}$-norm is uniformly bounded on any compact interval. Then standard compact embedding results will imply the convergence (upon possibly passing to a subsequence) w.r.t. $C^{0,\alpha}$ on any compact interval, c.f. \cite[8.6]{Alt12}. The key observation is that $\dot{v}_{\omega}(t)\cdot v_{\omega}(t)=0$ for all $t$ and $\omega$ by the defining ODE (\ref{E1}) so that in conclusion $|v_{\omega}(t)|^2=|v_0|^2$ for all $t$ and $\omega$.
	
	We therefore find $|x_{\omega}(t)-x_0|\leq |v_0|t$ and similarly $|\dot{x}_{\omega}(t)|=|v_{\omega}(t)|=|v_0|$ and thus the $C^{0,1}$-norm of $x_{\omega}$ is bounded independently of $\omega$ on any compact interval. Hence, by compactness of the embedding $C^{0,1}([0,T])\hookrightarrow C^{0,\alpha}([0,T])$ for all $0\leq \alpha <1$ and any fixed $T>0$ we see that every subsequence of $x_{\omega_n}$ has yet another subsequence converging in $C^{0,\alpha}_{\operatorname{loc}}(\mathbb{R})$, i.e. converging on any compact interval, to some limit function $x(t)$.
	\newline
	\newline
	\underline{Step 2:} In this step we will show that if $x_{\omega_n}$ with $\omega_n\rightarrow \infty$ converges to some limit function, then it us uniquely characterised by a certain 2nd order ODE, so that the limit will turn out to be independent of the sequence $(\omega_n)_n$, completing the convergence proof.
	
	We start by setting $y_{\omega}(t):=\int_0^tv^\perp_{\omega}(s)ds$ and show now that $y_{\omega}(t)\rightarrow 0$ in $C^0_{\operatorname{loc}}$. To see this, we observe that we have the identity
	\begin{gather}
		\label{S2E2}
		v^\perp_{\omega}(s)=b(x_{\omega}(s))\times (v_{\omega}(s)\times b(x_{\omega}(s)))=\frac{b(x_{\omega}(s))\times \dot{v}_{\omega}(s)}{\omega |B(x_{\omega}(s))|}
	\end{gather}
	where $b(y)=\frac{B(y)}{|B(y)|}$. We can therefore perform an integration by parts and find
	\begin{gather}
		\nonumber
		y_{\omega}(t)=\frac{b(x_{\omega}(t))\times v_{\omega}(t)}{\omega |B(x_{\omega}(t))|}-\frac{b(x_0)\times v_0}{\omega|B(x_0)|}-\frac{1}{\omega}\int_0^t\frac{d}{ds}\left(\frac{b(x_{\omega}(s))}{|B(x_{\omega}(s))|}\right)\times v_{\omega}(s)ds.
	\end{gather}
	Now, $|b(x_{\omega}(t))|=1$ for all $t,\omega$ since $b$ is a unit field. Further, $|v_{\omega}(t)|=|v_0|$ for all $t$ and $\omega$. Then on the one hand, since $|x_{\omega}(t)|\leq |x_0|+|v_0|t$ it follows immediately that $\frac{1}{|B(x_{\omega}(t))|}$ can be bounded above by some constant depending on $t$, but independent of $\omega$. On the other hand, if in addition $B$ has a minimal growth of order $\gamma$ at infinity we find for $|x_{\omega}(t)|\leq R$, $|B(x_{\omega}(t))|\geq \min_{y\in B_R(0)}|B(y)|=:\frac{1}{c}$ and for $|x_{\omega}(t)|\geq R$, $|B(x_{\omega}(t))|\geq \frac{C}{|x_{\omega}(t)|^\gamma}$ so that in the latter case $\frac{C}{|B(x_{\omega}(t))|}\leq |x_{\omega}(t)|^{\gamma}\leq (|v_0|t+|x_0|)^{\gamma}$ where we used our standard estimate $|x_{\omega}(t)-x_0|\leq |v_0|t$. Overall $\left|\frac{b(x_{\omega}(t))\times v_{\omega}(t)}{|B(x_{\omega}(t))|}\right|\leq \hat{C}(1+t^{\gamma})$ for a suitable constant $\hat{C}$ which is independent of $t$ and $\omega$. From now on we assume w.l.o.g. that $B$ grows polynomially of order $\gamma$ at infinity, since otherwise we can always use the same local argument as we presented just now with a constant depending on the upper threshold $T$ of any fixed interval $[0,T]$.
	
	As for the integrand, we observe that $\frac{d}{ds}\frac{b(x_{\omega})}{|B(x_{\omega})|}=\frac{Db(x_{\omega}(s))\cdot v_{\omega}(s)}{|B(x_{\omega}(s))|}-\frac{b(x_{\omega}(s))}{|B(x_{\omega}(s))|^3}\nabla \frac{|B|^2}{2}\cdot v_{\omega}(s)$ and that further $Db(y)=\frac{DB(y)}{|B(y)|}+(\nabla |B|^{-1})\cdot B^{\operatorname{Tr}}(y)$ so that, using the fact that $B$ is globally Lipschitz continuous, we find that
	\begin{gather}
		\nonumber
		\left|\frac{d}{ds}\frac{b(x_{\omega}(s))}{|B(x_{\omega}(s))|}\right|\leq \frac{C}{|B(x_{\omega}(t))|^2}\leq \hat{C}(|v_0|t+|x_0|)^{2\gamma}
	\end{gather}
	for constants $C,\hat{C}$ which are independent of $t$ and $\omega$ and so consequently
	\begin{gather}
		\label{S2E3}
		|y_{\omega}(t)|\leq \frac{C}{\omega}\left(1+t^{2\gamma+1}\right)\text{ for all }t\geq 0
	\end{gather}
	where $C>0$ is some constant which is independent of $t$ and $\omega$.
	\newline
	\newline
	We now turn to the remaining term which we abbreviate
	\begin{gather}
		\nonumber
		z_{\omega}(t):=x_0+\int_0^tv^\parallel_{\omega}(s)ds=x_0+\int_0^th_{\omega}(s)b(x_{\omega}(s))ds
	\end{gather}
	where we recall that $h_{\omega}(s)=v_{\omega}(s)\cdot b(x_{\omega}(s))$ and that we assume that we are given a subsequence of a subsequence of $(\omega_n)_n$ such that $x_{\omega}$ converges to a limit function $x$ in $C^{0,\alpha}_{\operatorname{loc}}(\mathbb{R})$ for every $0\leq \alpha <1$. We show that under these circumstances $h_{\omega}(t)$ will also converge in $C^{0,\alpha}_{\operatorname{loc}}(\mathbb{R})$ after possibly passing to yet another subsequence. To see this we notice first that $|h_{\omega}(s)|\leq |v_0|$ where we used that $b$ is a unit field and that the speed is preserved in time. Further, we compute
	\begin{gather}
		\nonumber
		\left|\frac{d}{ds}h_{\omega}(s)\right|=\left|(Db)(x_{\omega}(s))\cdot v_{\omega}(s)\right|\leq \left|Db(x_{\omega}(s))\right||v_0|
	\end{gather}
	where we used that $\dot{v}_{\omega}(s)$ is perpendicular to $b(x_{\omega}(s))$ according to (\ref{E1}) and that speed is preserved in time. Lastly, we observe that $|x_{\omega}(s)|\leq |x_0|+|v_0|s\leq |x_0|+|v_0|t$ for all $0\leq s\leq t$ is bounded on any compact interval so that, by continuity of $Db$ it follows that $|Db(x_{\omega}(s))|\leq c$ for some constant $c$ which may depend on $t$. This shows that the $C^{0,1}$-norm of $h_{\omega}$ is bounded on any compact interval so that (after possibly passing to a subsequence) we may assume that $h_{\omega}(t)$ converges to some limit function $h(t)$.
	
	This allows us to define $v(s):=h(s)b(x(s))$ where $h$ is the limit function of our subsequence of $h_{\omega_n}$ and $x$ is the limit trajectory of the corresponding subsequence of $x_{\omega_n}$. With this definition and using that $y_{\omega}(t)\rightarrow 0$ in $C^0_{\operatorname{loc}}(\mathbb{R})$, see (\ref{S2E3}), we see from the definition of $z_{\omega}$ and the relation $z_{\omega}=x_{\omega}-y_{\omega}$ that
	\begin{gather}
		\nonumber
		x(t)=x_0+\int_0^tv(s)ds=x_0+\int_0^th(s)b(x(s))ds
	\end{gather}
	which already shows that if the limit trajectory is unique, it must follow a magnetic field line. Further, since the $C^{0,1}$-norm of the $h_{\omega}$ was bounded we conclude that $h$ is locally Lipschitz continuous and the same is true for $x_{\omega}$ so that also $x$ is locally $C^{0,1}$. The above identity allows us to conclude at this point that $x$ is locally $C^{1,1}$.
	
	Our goal now will be to verify that the pair $(h,x)$ satisfies a system of first order ODEs and infer from that that these limit functions are in fact independent of the chosen subsequences. We set $h_0:=(v_0\cdot b(x_0))=h_{\omega}(0)$ and write
	\begin{gather}
		\label{S2E4}
		h_{\omega}(t)-h_0=\int_0^t\frac{d}{ds}h_{\omega}(s)ds=\int_0^tv^{\operatorname{Tr}}_{\omega}(s)\cdot Db(x_{\omega}(s))\cdot v_{\omega}(s)ds
	\end{gather}
	where we used again that $\dot{v}_{\omega}$ is perpendicular to $b(x_{\omega}(s))$. We observe now that $b^{\operatorname{Tr}}(y)\cdot Db(y)=\nabla |b|^2=0$ because $b$ is a unit field and hence
	\begin{gather}
		\nonumber
		h_{\omega}(t)-h_0=\int_0^tv^{\perp}_{\omega}(s)\cdot \left(Db(x_{\omega}(s))\cdot v_{\omega}(s)\right)ds.
	\end{gather}
	We use once more that we can express $v^\perp_{\omega}(s)=\frac{b(x_{\omega}(s))\times \dot{v}_{\omega}(s)}{\omega |B(x_{\omega}(s))|}$, (\ref{S2E2}), and get
	\begin{gather}
		\nonumber
		h_{\omega}(t)-h_0=\frac{1}{\omega}\int_0^t\dot{v}_{\omega}(s)\cdot \left(\frac{(Db(x_{\omega}(s))\cdot v_{\omega}(s))\times b(x_{\omega}(s))}{ |B(x_{\omega}(s))|}\right)ds.
	\end{gather}
	We integrate by parts and find
	\begin{gather}
		\nonumber
		h_{\omega}(t)-h_0=\frac{1}{\omega}v_{\omega}(s)\cdot \frac{Db(x_{\omega}(s))\cdot v_{\omega}(s)\times b(x_{\omega}(s))}{ |B(x_{\omega}(s))|}\mid_0^t
		\\
		\nonumber
		-\frac{1}{\omega}\int_0^tv_{\omega}(s)\cdot \frac{d}{ds}\left(\frac{(Db(x_{\omega}(s))\cdot v_{\omega}(s))\times b(x_{\omega}(s))}{ |B(x_{\omega}(s))|}\right)ds.
	\end{gather}
	We note that $\frac{d}{ds}b(x_{\omega}(s))=Db(x_{\omega})\cdot v_{\omega}(s)$ so that one term cancels immediately in the integrand. The remaining terms can be handled similarly as in the case of the estimate of $y_{\omega}$, (\ref{S2E3}), with the exception where the derivative in the integrand acts upon $v_{\omega}$, in which case we use once more the defining equation of $x_{\omega}$ (\ref{E1}). We then arrive at the following result
	\begin{gather}
		\nonumber
		h_{\omega}(t)-h_0=\int_0^t(v_{\omega}(s)\times b(x_{\omega}(s)))\cdot (Db(x_{\omega}(s))\cdot (v_{\omega}(s)\times b(x_{\omega}(s))))ds+\frac{\mathcal{O}(1)}{\omega}(1+t^{3\gamma+1})
	\end{gather}
	where $\mathcal{O}(1)$ is a term which can be bounded above by some constant which is independent of $t$ and $\omega$. More precisely, we have established the identity
	\begin{gather}
		\nonumber
		\int_0^tv^{\operatorname{Tr}}_{\omega}(s)\cdot Db(x_{\omega}(s))\cdot v_{\omega}(s)ds
		\\
		\label{S2E5}
		=\int_0^t(v_{\omega}(s)\times b(x_{\omega}(s)))\cdot (Db(x_{\omega}(s))\cdot (v_{\omega}(s)\times b(x_{\omega}(s))))ds+\frac{\mathcal{O}(1)}{\omega}(1+t^{3\gamma+1}).
	\end{gather}
	We will now make use of the following result whose proof we postpone to after the completion of the proof of \Cref{IT1}.
	\begin{lem}
		\label{S2L2}
		Let $b\in C^1(\mathbb{R}^3,\mathbb{R}^3)$ be a unit vector field, i.e. $|b(x)|=1$ for all $x\in \mathbb{R}^3$, and let $v\in \mathbb{R}^3$ be fixed, then
		\begin{gather}
			\nonumber
			v^{\operatorname{Tr}}\cdot Db(x)\cdot v+(v\times b(x))^{\operatorname{Tr}}\cdot Db(x)\cdot (v\times b(x))
			\\
			\nonumber
			=\operatorname{div}(b)(x)(|v|^2-(v\cdot b(x))^2)-(v\cdot b(x))(v\times b(x))\cdot \operatorname{curl}(b)(x)\text{ for all }x\in \mathbb{R}^3.
		\end{gather}
	\end{lem}
	We utilise \Cref{S2L2} with $v=v_{\omega}(s)$ and $x=x_{\omega}(s)$ and find, in combination with (\ref{S2E5})
	\begin{gather}
		\nonumber
		2\int_0^tv^{\operatorname{Tr}}_{\omega}(s)\cdot Db(x_{\omega}(s))\cdot v_{\omega}(s)ds
		\\
		\nonumber
		=\int_0^t\operatorname{div}(b)(x_{\omega}(s))(|v_0|^2-h^2_{\omega}(s))-h_{\omega}(s)(v_{\omega}(s)\times b(x_{\omega}(s)))\cdot \operatorname{curl}(b)(x_{\omega}(s))ds+\frac{\mathcal{O}(1)}{\omega}(1+t^{3\gamma+1})
	\end{gather}
	where we used that speed is preserved in time and the definition $h_{\omega}(t)=v_{\omega}(t)\cdot b(x_{\omega}(t))$. We combine this with (\ref{S2E4}) and find the identity
	\begin{gather}
		\nonumber
		h_{\omega}(t)-h_0=\int_0^t\frac{|v_0|^2-h^2_{\omega}(s)}{2}\operatorname{div}(b)(x_{\omega}(s))ds
		\\
		\nonumber
		-\int_0^th_{\omega}(s)(v_{\omega}(s)\times b(x_{\omega}(s)))\cdot \frac{\operatorname{curl}(b)(x_{\omega}(s))}{2}ds+\frac{\mathcal{O}(1)}{\omega}(1+t^{3\gamma+1}).
	\end{gather}
	We observe now that $v_{\omega}(t)\times b(x_{\omega}(t))=\frac{\dot{v}_{\omega}(t)}{\omega |B(x_{\omega}(t))|}$. We can therefore perform yet another integration by parts and obtain similarly as before that
	\begin{gather}
		\label{S2E6}
		h_{\omega}(t)-h_0=\int_0^t\frac{|v_0|^2-h^2_{\omega}(t)}{2}\operatorname{div}(b)(x_{\omega}(s))ds+\frac{\mathcal{O}(1)}{\omega}(1+t^{3\gamma+1}).
	\end{gather}
	We recall that we may assume that the $h_{\omega}$ and $x_{\omega}$ converge in $C^{0,\alpha}_{\operatorname{loc}}$ to a $C^{0,1}$ and $C^{1,1}$ function $h$ and $x$ respectively and that our goal is to show that these limit functions admit a characterisation which uniquely determines them through the choice of initial conditions $x_0,v_0$ and the vector field $B$ alone. In order to achieve this we will derive a second order ODE for $x(t)$ whose coefficients and initial conditions will be uniquely determined by $B,x_0,v_0$ alone which will establish the uniqueness of the limit. This will also imply that $h(t)=\dot{x}(t)\cdot b(x(t))$ is uniquely determined. We observe first that we may perform the limit $\omega\rightarrow\infty$ in (\ref{S2E6}) to find the pointwise identity
	\begin{gather}
		\label{S2E7}
		h(t)=h_0+\int_0^t\frac{|v_0|-h^2(s)}{2}\operatorname{div}(b)(x(s))ds.
	\end{gather}
	We notice first that (\ref{S2E7}) implies that $h\in C^{1,1}_{\operatorname{loc}}$ since we already know that $h\in C^{0,1}_{\operatorname{loc}}$ and $x\in C^{1,1}_{\operatorname{loc}}$. From here we can bootstrap the regularity by recalling that $\dot{x}(t)=h(t)b(x(t))\in C^{1,1}_{\operatorname{loc}}$ and thus $x\in C^{2,1}_{\operatorname{loc}}$ so that in turn by (\ref{S2E7}) $h\in C^{2,1}_{\operatorname{loc}}$ and by induction $h(t),x(t)$ are $C^{\infty}$-smooth whenever $B$ is. Now define the magnetic moment $\mu(t):=\frac{|v_0|^2-h^2(t)}{2|B(x(t))|}$. We claim that $\mu(t)$ is preserved in time. To see this we simply compute the derivative
	\begin{gather}
		\nonumber
		2\frac{d}{dt}\mu(t)=-2h(t)\frac{\dot{h}(t)}{|B(x(t))|}+(|v_0|^2-h^2(t))\frac{d}{dt}\left(|B|^{-1}(x(t))\right).
	\end{gather}
	We make use of (\ref{S2E7}) to see that $2\dot{h}(t)=(|v_0|^2-h^2(t))\operatorname{div}(b)(x(t))$ and further we note that $\operatorname{div}(b)=B\cdot \nabla \left(|B|^{-1}\right)$ by the product rule and because $B$ is div-free so that
	\begin{gather}
		\nonumber
		2\frac{d}{dt}\mu(t)=(|v_0|^2-h^2(t))\left(-h(t)b(x(t))\cdot (\nabla |B|^{-1})(x(t))+\frac{d}{dt}\left(|B|^{-1}(x(t))\right)\right).
	\end{gather}
	Lastly we recall that we have already established that $\dot{x}(t)=h(t)b(x(t))$ and so the chain rule yields $\frac{d}{dt}|B|^{-1}(x(t))=h(t)b(x(t))\cdot \nabla \left(|B|^{-1}\right)(x(t))$ and consequently $\mu(t)=\mu(0)=:\mu_0$ for all times $t$. This also proves \Cref{IC2}.
	
	With this in mind we can find the following equivalent expression for (\ref{S2E7})
	\begin{gather}
		\label{S2E8}
		h(t)=h_0+\mu_0\int_0^t|B(x(s))|\operatorname{div}(b)(x(s))ds.
	\end{gather}
	With this at hand we can compute
	\begin{gather}
		\nonumber
		\ddot{x}(t)=\frac{d}{dt}(h(t)b(x(t)))=\dot{h}(t)b(x(t))+h(t)Db(x(t))\cdot \dot{x}(t)
		\\
		\nonumber
		=\mu_0|B(x(t))|\operatorname{div}(b)(x(t))b(x(t))+h^2(t)Db(x(t))\cdot b(x(t))
		\\
		\nonumber
		=\mu_0|B(x(t))|\operatorname{div}(b)(x(t))b(x(t))+(|v_0|^2-2\mu_0|B(x(t))|)Db(x(t))\cdot b(x(t))
	\end{gather}
	where we used that $\mu$ is preserved in time. The corresponding initial conditions for this second order ODE are $x(0)=x_0$ and $\dot{x}(0)=h_0b(x_0)$ so that the limit function turns out to be the unique solution to the above 2nd order ODE with initial conditions $x(0)=x_0$ and $\dot{x}(0)=h_0b(x_0)$ which overall shows that every subsequence of $x_{\omega}$ admits yet another subsequence which converges to a limit trajectory which itself is independent of the chosen subsequence. This implies the convergence of the original sequence to the unique solution of the above second order ODE.
	\newline
	\newline
	\underline{Step 3:} Here we prove the claimed convergence rate. We start by expressing once more $x_{\omega}(t)=x_0+\int_0^t h_{\omega}(s)b(x_{\omega}(s))ds+\int_0^tv^\perp_{\omega}(s)ds$ and the limit trajectory by $x(t)=x_0+\int_0^t h(s)b(x(s))ds$. We recall that we had defined $y_{\omega}(t)=\int_0^t v^\perp_{\omega}(s)ds$ and found an estimate for $y_{\omega}(t)$ in (\ref{S2E3}) in the previous step. We hence find
	\begin{gather}
		\nonumber
		x_{\omega}(t)-x(t)=\int_0^t h_{\omega}(s)b(x_{\omega}(s))-h(s)b(x(s))ds+y_{\omega}(t)
		\\
		\nonumber
		=\int_0^t h_{\omega}(s)(b(x_{\omega}(s))-b(x(s)))ds+\int_0^t(h_{\omega}(s)-h(s))b(x(s))ds+y_{\omega}(t)
	\end{gather}
	and are left with estimating the two integral terms in the last line. We observe now that $|h_{\omega}(s)|\leq |v_0|$ and that
	\begin{gather}
		\nonumber
		b(x_{\omega}(s))-b(x(s))=\int_0^1 \frac{d}{d\tau}\left(b(x(s)+\tau (x_{\omega}(s)-x(s)))\right)d\tau
		\\
		\nonumber
		=\int_0^1(Db)(x(s)+\tau (x_{\omega}(s)-x(s)))d\tau \cdot (x_{\omega}(s)-x(s)).
	\end{gather}
	We now make use of the fact that $B$ is globally Lipschitz continuous and the growth properties of $B$ to deduce that $|Db(y)|\leq C(1+|y|^\gamma)$ for some $C>0$ independent of $y$ and thus $|Db(x(s)+\tau(x_{\omega}(s)-x(s))|\leq C(1+\left((1-\tau)|x(s)|+\tau|x_{\omega}(s)|\right)^{\gamma})$ and we can further use our standard estimates $|x(s)|\leq |x_0|+|v_0|s$ as well as $|x_{\omega}(s)|\leq |x_0|+|v_0|s$ where we used that $|v_{\omega}(s)|=|v_0|$ and that $|h(s)|$ is bounded above by $|v_0|$. We hence obtain the estimate
	\begin{gather}
		\nonumber
		|x_{\omega}(t)-x(t)|\leq C\frac{1+t^{2\gamma+1}}{\omega}+\int_0^t|h_{\omega}(s)-h(s)|ds+C(1+t^{\gamma+1})\int_0^t|x_{\omega}(s)-x(s)|ds
	\end{gather} 
	where $C>0$ is a constant which does not depend on $t$ and $\omega$. We apply now Gronwall's inequality to deduce
	\begin{gather}
		\label{S2E9}
		|x_{\omega}(t)-x(t)|\leq C\left(\frac{1+t^{2\gamma+1}}{\omega}+\int_0^t|h_{\omega}(s)-h(s)|ds\right)\exp\left(C t^{\gamma+2}\right)
	\end{gather}
	for some, possibly larger, constant $C>0$ which is independent of $t$ and $\omega$.
	\newline
	We now turn to the estimate of $\int_0^t|h_{\omega}(s)-h(s)|ds$. To this end we utilise identities (\ref{S2E6}),(\ref{S2E7}) and write
	\begin{gather}
		\nonumber
		2(h_{\omega}(t)-h(t))=\int_0^t(|v_0|^2-h^2_{\omega}(s))\operatorname{div}(b)(x_{\omega}(s))-(|v_0|^2-h^2(s))\operatorname{div}(b)(x(s))ds+\frac{\mathcal{O}(1)}{\omega}(1+t^{3\gamma+1})
		\\
		\nonumber
		=\int_0^t (\operatorname{div}(b_{\omega}(s))-\operatorname{div}(b(x(s))))(|v_0|^2-h^2_{\omega}(s))ds+\int_0^t\operatorname{div}(b)(x(s))\left(h(s)-h_{\omega}(s)\right)(h(s)+h_{\omega}(s))ds
		\\
		\nonumber
		+\frac{\mathcal{O}(1)}{\omega}(1+t^{3\gamma+1}).
	\end{gather}
	As previously in the case of $Db$ we obtain the estimate $|\operatorname{div}(b)(x(s))|\leq C(1+s^{\gamma})$ and we can use that $|h(s)|$ and $|h_{\omega}(s)|$ are bounded above by $|v_0|$. Further, we can argue as in the case of $b(x_{\omega}(s))-b(x(s))$ that
	\begin{gather}
		\nonumber
		|\operatorname{div}(b)(x_{\omega}(s))-\operatorname{div}(b)(x(s))|\leq C(1+t^{2\gamma+1}) |x_{\omega}(s)-x(s)|
	\end{gather}
	and so we obtain the corresponding estimate
	\begin{gather}
		\nonumber
		|h_{\omega}(t)-h(t)|\leq C(1+t^{\gamma})\int_0^t|h_{\omega}(s)-h(s)|ds+C(1+t^{2\gamma+1})\int_0^t|x_{\omega}(s)-x(s)|ds+\frac{\mathcal{O}(1)}{\omega}\left(1+t^{3\gamma+1}\right).
	\end{gather}
	A further application of Gronwall's inequality yields
	\begin{gather}
		\label{S2E10}
		|h_{\omega}(t)-h(t)|\leq C\left(\frac{1+t^{3\gamma+1}}{\omega}+(1+t^{2\gamma+1})\int_0^t|x_{\omega}(s)-x(s)|ds\right)\exp(Ct^{\gamma+1})
	\end{gather}
	for a suitable constant $C>0$ which is independent of $t$ and $\omega$. We now claim that there is some $0<C(B,x_0,v_0)<\infty$ such that
	\begin{gather}
		\label{S2E11}
		|x_{\omega}(t)-x(t)|\leq \frac{C}{\omega}\exp(C\exp(Ct^{\gamma+2}))\text{ for all }0\leq t<\infty\text{ and all }\omega>0.
	\end{gather}
	To see this, we define first for any given $f:[0,\infty)\rightarrow [0,\infty)$ the set
	\begin{gather}
		\label{S2E12}
		K_f:=\{T\geq 0\mid |x_{\omega}(t)-x(t)|\leq \frac{f(t)}{\omega}\text{ for all }0\leq t \leq T\text{ and all }\omega>0\}
	\end{gather}
	and $\tau_f:=\sup K_f$. We note that the claimed convergence rate is equivalent to the statement that $\tau_f=+\infty$ for the right choice of $f$. We will argue by contradiction by assuming first that for a given $f$ we have $\tau_f<\infty$ and then derive a certain inequality which then must hold true for any given such $f$. Then, we make a specific choice for $f$ which will lead to a contradiction and hence provide the desired convergence rate. 
	First, we observe that $0\in K_f$ for every $f$ since $x_{\omega}(0)=x_0=x(0)$.
	
	Now let us suppose that $0\leq \tau_f<\infty$ for some chosen function $f$. To simplify the argument we note that (\ref{S2E9}) and (\ref{S2E10}) imply the following rougher upper bound with a possibly larger constant $C$
	\begin{gather}
		\nonumber
		|x_{\omega}(t)-x(t)|\leq C\exp(Ct^{\gamma+2})\left(\frac{1}{\omega}+\int_0^t|h_{\omega}(s)-h(s)|ds \right)\text{, }
		\\
		\label{S2E13}
		|h_{\omega}(t)-h(t)|\leq C\exp(Ct^{\gamma+1})\left(\frac{1}{\omega}+\int_0^t|x_{\omega}(s)-x(s)|ds\right).
	\end{gather}
	Now we fix $0<\epsilon<1$ and see that for $\tau\leq t<\tau+\epsilon$ (we drop the subscript $f$ to simplify notation) we have by splitting $\int_0^t|x_{\omega}(s)-x(s)|ds=\int_0^\tau|x_{\omega}(s)-x(s)|ds+\int_{\tau}^t|x_{\omega}(s)-x(s)|ds$ and making use of (\ref{S2E12}) and the definition of $\tau$
	\begin{gather}
		\label{S2E14}
		|h_{\omega}(t)-h(t)|\leq C\exp(Ct^{\gamma+1})\left(\frac{1+\int_0^\tau f(s)ds}{\omega}+\epsilon\|x_{\omega}-x\|_{C^0([\tau,t])}\right)\text{ for }\tau\leq t\leq \tau+\epsilon
	\end{gather}
	where we used (\ref{S2E13}). We note that if $0\leq t\leq \tau$ we obtain an identical estimate as in (\ref{S2E14}) with the exception that the last term in (\ref{S2E14}) will be missing, i.e.
	\begin{gather}
		\nonumber
		|h_{\omega}(t)-h(t)|\leq C\exp(Ct^{\gamma+1})\frac{1+\int_0^\tau f(s)ds}{\omega}\text{ for }0\leq t\leq \tau
	\end{gather}
	which we can insert into (\ref{S2E13}) to deduce for $\tau\leq t\leq \tau+\epsilon$
	\begin{gather}
		\nonumber
		|x_{\omega}(t)-x(t)|\leq C\exp(Ct^{\gamma+2})\left[\frac{1}{\omega}+\frac{C}{\omega}\int_0^t\exp(Cs^{\gamma+1})ds\left(1+\int_0^{\tau}f(s)ds\right)\right]
		\\
		\nonumber
		+C\epsilon^2(\tau+1)\exp\left(C ((\tau+1)^{\gamma+2}+(\tau+1)^{\gamma+1})\right)\|x_{\omega}-x\|_{C^0([\tau,t])}\text{ for all }\tau\leq t\leq \tau+\epsilon
	\end{gather}
	for a possibly yet again larger constant $C>0$, where we used that $\epsilon<1$. We may then take the supremum over $[\tau,T]$ for any fixed $\tau\leq T\leq \tau+\epsilon$ and see that for suitably small $0<\epsilon$ we find
	\begin{gather}
		\nonumber
		|x_{\omega}(T)-x(T)|\leq \|x_{\omega}-x\|_{C^0([\tau,T])}\leq \frac{C\exp(CT^{\gamma+2})}{\omega}\left(1+\int_0^T\exp(Cs^{\gamma+1})ds\left(1+\int_0^{\tau}f(s)ds\right)\right).
	\end{gather}
	We now observe that since $f\geq0$ we have $\int_0^{\tau}f(s)ds\leq \int_0^{t}f(s)ds$ for all $\tau\leq t$. Further, we notice that $\frac{\int_0^t\exp(Cs^{\gamma+1})ds}{\exp(Ct^{\gamma+1})}$ is globally bounded on $[0,\infty)$, where the global bound depends on $C$ and $\gamma$ so that with some additional simple estimates we arrive at the rough upper bound (with a possibly larger constant $C$)
	\begin{gather}
		\label{S2E15}
		|x_{\omega}(t)-x(t)|\leq \frac{C\exp(Ct^{\gamma+2})}{\omega}\left(1+\int_0^tf(s)ds\right)\text{ for all }\tau_f\leq t\leq \tau_f+\epsilon_f
	\end{gather}
	for a suitable constant $C$ which is independent of $f$,$t$, and $\omega$. We now make a specific choice for $f$, namely we define first $g(t):=C\exp(Ct^{\gamma+2})$ with $C$ and $\gamma$ as in (\ref{S2E15}), $h(t):=\frac{g^2(t)+g^{\prime}(t)}{g(t)}$ and finally
	\begin{gather}
		\label{S2E16}
		f(t):=g(0)\exp\left(\int_0^th(s)ds\right)
	\end{gather}
	which is a non-negative function. This function satisfies
	\begin{gather}
		\nonumber
		f^{\prime}(t)=h(t)f(t)=g(t)f(t)+f(t)\frac{g^{\prime}(t)}{g(t)}\Leftrightarrow f(t)=\frac{f^{\prime}(t)}{g(t)}-\frac{f(t)g^{\prime}(t)}{g^2(t)}=\left(\frac{f(t)}{g(t)}\right)^{\prime}.
	\end{gather}
	Integrating the above identity yields
	\begin{gather}
		\nonumber
		\int_0^tf(s)ds=\frac{f(t)}{g(t)}-1\Leftrightarrow 1+\int_0^tf(s)ds=\frac{f(t)}{g(t)}
	\end{gather}
	where we used that $f(0)=g(0)$. We can insert this identity in (\ref{S2E15}) and overall obtain the inequality
	\begin{gather}
		\nonumber
		|x_{\omega}(t)-x(t)|\leq \frac{f(t)}{\omega}\text{ for all }\tau_f\leq t\leq \tau_f+\epsilon_f
	\end{gather}
	for a suitable $\epsilon_f>0$. By definition of $K_f$, recall (\ref{S2E12}) the same inequality holds true for $0\leq t<\tau_f$ which implies that there is some $t_f>\tau_f$ with $t_f\in K_f$, contradicting the definition of $\tau_f$ as the supremum of $K_f$. We conclude that $\tau_f=+\infty$ for this specific choice of $f$ and hence, by definition of $K_f$, we arrive at the inequality
	\begin{gather}
		\label{S2E17}
		|x_{\omega}(t)-x(t)|\leq \frac{f(t)}{\omega}\text{ for all }0\leq t<\infty\text{ and all }0<\omega<\infty
	\end{gather} 
	where $f(t)$ is given by (\ref{S2E16}). Finally, we observe that $g^{\prime}(t)=C(\gamma+2)t^{\gamma+1}g(t)$ so that we obtain a rough upper bound on $h(t)$ of the form $h(t)=\frac{g^2(t)+g^{\prime}(t)}{g(t)}\leq C\exp(Ct^{\gamma+2})$ with a possibly larger constant $C>0$. We can then use once more the fact that $\frac{\int_0^t\exp(Cs^{\gamma+2})ds}{\int_0^t\exp(Ct^{\gamma+2})}$ is globally bounded to arrive at the claimed estimate
	\begin{gather}
		\label{S2E18}
		|x_{\omega}(t)-x(t)|\leq \frac{C}{\omega}\exp\left(C\exp(Ct^{\gamma+2})\right)\text{ for all }t\geq 0\text{, }\omega>0
	\end{gather}
	for a suitable constant $C>0$ which is independent of $\omega$ and $t$.
	
	For future reference of the proof of \Cref{IT3} we note that a combination of (\ref{S2E18}) and (\ref{S2E13}) implies
	\begin{gather}
		\label{S2E19}
		|h_{\omega}(t)-h(t)|\leq \frac{C}{\omega}\exp(C\exp(Ct^{\gamma+2}))\text{ for all }t\geq 0\text{, }\omega>0
	\end{gather}
	upon increasing the constant $C>0$ if necessary.
\end{proof}
\begin{proof}[Proof of \Cref{S2L2}]
	We start by letting $z(x)$ be any $C^1$-vector field and we observe that $z^{\operatorname{Tr}}\cdot Db(x)\cdot z=z(x)\cdot \nabla_zb(x)$ where $\nabla_zb(x)=z^j(x)(\partial_jb^k)(x)e_k$ is the covariant derivative of $b$ along $z$ (we utilise Einstein's summation convention and $e_k$ denote the standard Euclidean basis vectors). We then recall the following calculus identity
	\begin{gather}
		\nonumber
		\nabla (z(x)\cdot b(x))=\nabla_zb+\nabla_bz+z\times \operatorname{curl}(b)+b\times \operatorname{curl}(z).
	\end{gather}
	Consequently $z^{\operatorname{Tr}}\cdot Db(x)\cdot z=z\cdot \nabla (b\cdot z)-z\cdot \nabla_bz-z\cdot (b\times \operatorname{curl}(z))$. By properties of the Levi-Civita connection (or a simple direct calculation) we see that $z\cdot \nabla_bz=b\cdot \nabla \frac{|z|^2}{2}$. Lastly, we now set for our fixed $v\in \mathbb{R}^3$, $z(x):=v\times b(x)$ which implies that $z\cdot b=0$ everywhere and thus $\nabla (z\cdot b)=0$ so that we find
	\begin{gather}
		\nonumber
		z^{\operatorname{Tr}}\cdot Db(x)\cdot z=-b\cdot \nabla \frac{|z|^2}{2}+z\cdot (\operatorname{curl}(z)\times b).
	\end{gather}
	To express $\operatorname{curl}(z)$ in a different way we make use of the identity $\operatorname{curl}(X\times Y)=X\operatorname{div}(Y)-Y\operatorname{div}(X)-[X,Y]$ valid for any two $C^1$-vector fields $X,Y$, where $[X,Y]=\nabla_XY-\nabla_YX$ is the Lie-bracket of $X,Y$. Recall that $z(x)=v\times b(x)$ where $v$ is just some fixed vector and hence independent of the variable $x$, so that we find $\operatorname{curl}(z)=v\operatorname{div}(b)-\nabla_vb$ where we used that $v$ is independent of $x$.
	
	Using the cyclic properties of the inner product, we find
	\begin{gather}
		\nonumber
		z\cdot (\operatorname{curl}(z)\times b)=\operatorname{div}(b)|z|^2+(z\times b)\cdot \nabla_vb.
	\end{gather}
	Using the triple product rule, we further find $z\times b=(v\times b)\times b=(b\cdot v)b-v|b|^2=(b\cdot v)b-v$ since $b$ is a unit field. We overall obtain
	\begin{gather}
		\nonumber
		z^{\operatorname{Tr}}\cdot Db(x)\cdot z=-b\cdot \nabla \frac{|z|^2}{2}+\operatorname{div}(b)|z|^2+(b\cdot v)b\cdot \nabla_vb-v\cdot \nabla_vb.
	\end{gather}
	We notice now that once more by properties of the Levi-Civita connection we have $b\cdot \nabla_vb=v\cdot \nabla \frac{|b|^2}{2}=0$ since $b$ is a unit field. Hence we find
	\begin{gather}
		\nonumber
		z^{\operatorname{Tr}}\cdot Db(x)\cdot z=-b\cdot \nabla \frac{|z|^2}{2}+\operatorname{div}(b)|z|^2-v\cdot \nabla_vb.
	\end{gather}
	Now we note that $|z|^2=|v\times b|^2=|v|^2-(v\cdot b)^2$ since $b$ is a unit field. We conclude that $\nabla |z|^2=-2(v\cdot b)\nabla (v\cdot b)$ where we used that $v$ is a fixed vector independent of $x$. Now, by properties of the Levi-Civita connection we find $b\cdot \nabla (v\cdot b)=b\cdot \nabla_bv+v\cdot \nabla_bb=v\cdot \nabla_bb$ because $v$ is independent of $x$. With this we arrive at
	\begin{gather}
		\nonumber
		z^{\operatorname{Tr}}\cdot Db(x)\cdot z+v\cdot \nabla_vb=\operatorname{div}(b)(x)|z|^2+(v\cdot b)v\cdot \nabla_bb.
	\end{gather}
	Finally, we make use of the vector calculus identity $\nabla_bb=\nabla \frac{|b|^2}{2}+\operatorname{curl}(b)\times b=\operatorname{curl}(b)\times b$ because $b$ is a unit field and arrive at
	\begin{gather}
		\nonumber
		z^{\operatorname{Tr}}\cdot Db(x)\cdot z+v\cdot \nabla_vb=\operatorname{div}(b)(x)|z|^2-(v\cdot b)(v\times b)\cdot\operatorname{curl}(b)
	\end{gather}
	by cyclic properties of the inner product. We are left with noticing that $v\cdot \nabla_vb=v^{\operatorname{Tr}}\cdot Db(x)\cdot v$, recalling that $z(x)=v\times b(x)$ and that $|z|^2=|v|^2-(v\cdot b(x))^2$ since $b$ is a unit field so that we arrive at the claimed identity.
\end{proof}
\section{Proof of \Cref{IT3}}
\begin{proof}[Proof of \Cref{IT3}]
	We start by expressing
	\begin{gather}
		\nonumber
		p(x_{\omega}(t))-p_0=p(x_{\omega}(t))-p(x_0)=p(x_{\omega}(t))-p(x_{\omega}(0))=\int_0^t\frac{d}{ds}p(x_{\omega}(s))ds=\int_0^t(\nabla p)(x_{\omega}(s))\cdot v_{\omega}(s)ds.
	\end{gather}
	We recall that we assume that $B$ satisfies the PDE
	\begin{gather}
		\label{S3E1}
		B\times \operatorname{curl}(B)=\nabla p\text{ in }\Omega\text{, }B,\operatorname{curl}(B)\parallel \partial\Omega
	\end{gather}
	for a domain $\Omega\subset\mathbb{R}^3$ and a suitable $\nabla p\in C^2(\overline{\Omega})$. We then define as usual $b(y):=\frac{B(y)}{|B(y)|},h_{\omega}(t):=v_{\omega}(t)\cdot b(x_{\omega}(t))$, $v^\parallel_{\omega}(t):=h_{\omega}(t)b(x_{\omega}(t))$,$v^{\perp}_{\omega}(t):=v_{\omega}(t)-v^{\parallel}_{\omega}(t)$. With these definitions and due to the relation (\ref{S3E1}) we arrive at
	\begin{gather}
		\nonumber
		p(x_{\omega}(t))-p_0=\int_0^t(\nabla p)(x_{\omega}(s))\cdot v^{\perp}_{\omega}(s)ds.
	\end{gather}
	Now we recall that we have the identity, see (\ref{S2E2}), $v^{\perp}_{\omega}(s)=\frac{b(x_{\omega}(s))\times \dot{v}_{\omega}(s)}{\omega|B(x_{\omega}(s))|}$ which we can insert and integrate by parts
	\begin{gather}
		\nonumber
		p(x_{\omega}(t))-p_0=\frac{\nabla p(x_{\omega}(t))\cdot (b(x_{\omega}(t))\times v_{\omega}(t))}{\omega|B(x_{\omega}(t))|}-\frac{\nabla p(x_0)\cdot (b(x_0)\times v_0)}{\omega |B(x_0)|}
		\\
		\label{S3E2}
		-\frac{1}{\omega}\int_0^tv_{\omega}(s)\cdot \frac{d}{ds}\left(\frac{\nabla p(x_{\omega}(s))\times b(x_{\omega}(s))}{|B(x_{\omega}(s))|}\right)ds.
	\end{gather}
	We compute
	\begin{gather}
		\label{S3E3}
		v_{\omega}(s)\cdot \frac{d}{ds}\left(\frac{\nabla p(x_{\omega}(s))\times b(x_{\omega}(s))}{|B(x_{\omega}(s))|}\right)=(b(x_{\omega}(s))\times v_{\omega}(s))^{\operatorname{Tr}}\cdot \frac{\operatorname{Hess}(p)(x_{\omega}(s))}{|B(x_{\omega}(s))|}\cdot v_{\omega}(s)
		\\
		\nonumber
		+(v_{\omega}(s)\times \nabla p(x_{\omega}(s)))^{\operatorname{Tr}}\cdot \frac{Db(x_{\omega}(s))}{|B(x_{\omega}(s))|}\cdot v_{\omega}(s)+v_{\omega}(s)\cdot (\nabla p(x_{\omega})(s)\times b(x_{\omega}(s)))(\nabla |B|^{-1})(x_{\omega}(s))\cdot v_{\omega}(s).
	\end{gather}
	We start with the Hessian and use the notation $b_{\omega}(s)\equiv b(x_{\omega}(s))$, $\nabla p_{\omega}(s)\equiv (\nabla p)(x_{\omega}(s))$ etc.
	\begin{gather}
		\nonumber
		\int_0^t((b_{\omega}(s)\times v_{\omega}(s)))^{\operatorname{Tr}}\cdot \frac{\operatorname{Hess}(p_{\omega})(s)}{|B_{\omega}(s)|}\cdot v_{\omega}(s)ds
		\\
		\nonumber
		=\int_0^t((b_{\omega}(s)\times v_{\omega}(s)))^{\operatorname{Tr}}\cdot \frac{\operatorname{Hess}(p_{\omega})(s)}{|B_{\omega}(s)|}\cdot v^{\parallel}_{\omega}(s)ds+\int_0^t((b_{\omega}(s)\times v_{\omega}(s)))^{\operatorname{Tr}}\cdot \frac{\operatorname{Hess}(p_{\omega})(s)}{|B_{\omega}(s)|}\cdot v^{\perp}_{\omega}(s)ds.
	\end{gather}
	We write again $v^{\perp}_{\omega}(s)=\frac{b_{\omega}(s)\times \dot{v}_{\omega}(s)}{\omega|B_{\omega}(s)|}$ and integrate by parts. We further make use of the fact that we only consider times $t$ up to which $x_{\omega}(t)$ stays inside of $\Omega$ and since $B$ is non-vanishing on $\mathbb{R}^3$ and $\overline{\Omega}$ is compact, we can estimate the term involving $|B(x_{\omega}(t))|$ from below and above by some constants depending only on $B$ and $\Omega$.
	
	\begin{rem}
		If we would want to deal with possibly unbounded domains $\Omega$, for instance $\Omega=\mathbb{R}^3$, then one would once again need to impose decay conditions on $B$ at infinity akin to \Cref{IT1} to obtain a corresponding convergence rate.
	\end{rem}
	Doing this we find
	\begin{gather}
		\nonumber
		\int_0^t((b_{\omega}(s)\times v_{\omega}(s)))^{\operatorname{Tr}}\cdot \frac{\operatorname{Hess}(p_{\omega})(s)}{|B_{\omega}(s)|}\cdot v^{\perp}_{\omega}(s)ds
		\\
		\nonumber
		=\frac{t\mathcal{O}(1)}{\omega}+\frac{1}{\omega}\int_0^t (b_{\omega}(s)\times \dot{v}_{\omega}(s)\cdot \frac{\operatorname{Hess}(p_{\omega}(s)))}{|B_{\omega}(s)|^2}\cdot (b_{\omega}(s)\times v_{\omega}(s))ds
	\end{gather}
	where $\mathcal{O}(1)$ is a term which can be bounded above by a constant, which depends only on $B,p$ and $\Omega$ but is independent of $t$ and $\omega$, where as usual the time $t$ up to which this estimate is valid coincides with the minimal first exit time of $x_{\omega}(t)$ out of $\Omega$.
	
	We make now use of the defining equation $\dot{v}_{\omega}=\omega v_{\omega}\times B_{\omega}$ and then arrive at the identity
	\begin{gather}
		\nonumber
		\int_0^t((b_{\omega}(s)\times v_{\omega}(s)))^{\operatorname{Tr}}\cdot \frac{\operatorname{Hess}(p_{\omega})(s)}{|B_{\omega}(s)|}\cdot v^{\perp}_{\omega}(s)ds+\int_0^t(v^{\perp}_{\omega}(s))^{\operatorname{Tr}}\cdot \frac{\operatorname{Hess}(p_{\omega}(s))}{|B_{\omega}(s)|}\cdot (b_{\omega}(s)\times v_{\omega}(s))ds=\frac{t\mathcal{O}(1)}{\omega}.
	\end{gather}
	Lastly, we can use the fact that the Hessian is symmetric to deduce that
	\begin{gather}
		\nonumber
		\int_0^t((b_{\omega}(s)\times v_{\omega}(s)))^{\operatorname{Tr}}\cdot \frac{\operatorname{Hess}(p_{\omega})(s)}{|B_{\omega}(s)|}\cdot v^{\perp}_{\omega}(s)ds=\frac{t\mathcal{O}(1)}{\omega}.
	\end{gather}
	Using the fact that $b_{\omega}(s)\times v_{\omega}(s)=-\frac{\dot{v}_{\omega}(s)}{\omega}$, that $v^{\parallel}_{\omega}(s)=h_{\omega}(s)b(x_{\omega}(s))$, that $|v_{\omega}(s)|=|v_0|$ and that, as we showed in the proof of \Cref{IT1}, $|\dot{h}_{\omega}(s)|$ is uniformly bounded it follows easily that
	\begin{gather}
		\nonumber
		\int_0^t((b_{\omega}(s)\times v_{\omega}(s)))^{\operatorname{Tr}}\cdot \frac{\operatorname{Hess}(p_{\omega})(s)}{|B_{\omega}(s)|}\cdot v^{\parallel}_{\omega}(s)ds=\frac{t\mathcal{O}(1)}{\omega}
	\end{gather}
	and so overall we find
	\begin{gather}
		\label{S3E4}
		\int_0^t((b_{\omega}(s)\times v_{\omega}(s)))^{\operatorname{Tr}}\cdot \frac{\operatorname{Hess}(p_{\omega})(s)}{|B_{\omega}(s)|}\cdot v_{\omega}(s)ds=\frac{t\mathcal{O}(1)}{\omega}.
	\end{gather}
	We now turn to the term in (\ref{S3E3}) which contains the Jacobian of $b$
	\begin{gather}
		\nonumber
		\int_0^t(v_{\omega}(s)\times \nabla p_{\omega}(s))^{\operatorname{Tr}}\cdot \frac{Db_{\omega}(s)}{|B_{\omega}(s)|}\cdot v_{\omega}(s)ds.
	\end{gather}
	We decompose as usual $v_{\omega}=v^{\parallel}_{\omega}+v^{\perp}_{\omega}$ and as in the case of the Hessian arrive at an identity of the form
	\begin{gather}
		\nonumber
		\int_0^t(v_{\omega}(s)\times \nabla p_{\omega}(s))^{\operatorname{Tr}}\cdot \frac{Db_{\omega}(s)}{|B_{\omega}(s)|}\cdot v_{\omega}(s)ds=\frac{t\mathcal{O}(1)}{\omega}+\int_0^t(v^{\parallel}_{\omega}(s)\times \nabla p_{\omega}(s))^{\operatorname{Tr}}\cdot \frac{Db_{\omega}(s)}{|B_{\omega}(s)|}\cdot v^{\parallel}_{\omega}(s)ds
		\\
		\nonumber
		+\int_0^t (v^{\perp}_{\omega}(s)\times \nabla p_{\omega}(s))\cdot \frac{Db_{\omega}(s)}{|B_{\omega}(s)|}\cdot v_{\omega}(s)ds.
	\end{gather}
	To further simplify the last term, we utilise once more the identity $v^{\perp}_{\omega}=\frac{b_{\omega}\times \dot{v}_{\omega}}{\omega|B_{\omega}|}$ which leads us to the identity
	\begin{gather}
		\nonumber
		\int_0^t (v^{\perp}_{\omega}(s)\times \nabla p_{\omega}(s))\cdot \frac{Db_{\omega}(s)}{|B_{\omega}(s)|}\cdot v_{\omega}(s)ds=\frac{\mathcal{O}(1)t}{\omega}
		\\
		\nonumber
		+\int_0^t\frac{\nabla p_{\omega}(s)\cdot (Db_{\omega}(s)\cdot (v_{\omega}(s)\times b_{\omega}(s))\times (v_{\omega}(s)\times b_{\omega}(s)))}{|B_{\omega}(s)|}ds.
	\end{gather}
	We will now make use of the following Lemma whose proof we give after the completion of the proof of \Cref{IT3}
	\begin{lem}
		\label{S3L2}
		Let $v\in \mathbb{R}^3$ and $b:\mathbb{R}^3\rightarrow\mathbb{R}^3$ be a unit vector field. Then for every $x\in \mathbb{R}^3$:
		\begin{gather}
			\nonumber
			(Db(x)\cdot (v\times b(x)))\times (v\times b(x))
			\\
			\nonumber
			=\left((b\cdot v)\cdot (v\cdot \operatorname{curl}(b)(x))-|v|^2(b(x)\cdot \operatorname{curl}(b)(x))-\nabla_vb(x)\cdot (v\times b(x))\right)b(x).
			\end{gather}
	\end{lem}
	Using \Cref{S3L2} and keeping in mind that $\nabla p$ and $b$ are everywhere orthogonal we conclude that
	\begin{gather}
		\nonumber
		\int_0^t (v^{\perp}_{\omega}(s)\times \nabla p_{\omega}(s))\cdot \frac{Db_{\omega}(s)}{|B_{\omega}(s)|}\cdot v_{\omega}(s)ds=\frac{\mathcal{O}(1)t}{\omega}.
	\end{gather}
	Using further that $v^{\parallel}_{\omega}(s)=h_{\omega}(s)b_{\omega}(s)$ and that $Db_{\omega}\cdot b_{\omega}=\nabla_{b}b(x_{\omega}(s))$ we may insert all our previous findings into (\ref{S3E2}) to deduce that
	\begin{gather}
		\nonumber
		p(x_{\omega}(t))-p_0=\frac{\nabla p(x_{\omega}(t))\cdot (b(x_{\omega}(t))\times v_{\omega}(t))}{\omega|B_{x_{\omega}(t)}|}-\frac{\nabla p(x_0)\cdot (b(x_0)\times v_0)}{\omega|B(x_0)|}
		\\
		\nonumber
		-\frac{1}{\omega}\int_0^tv_{\omega}(s)\cdot (\nabla p(x_{\omega}(s))\times b(x_{\omega}(s)))(\nabla |B|^{-1})(x_{\omega}(s))\cdot v_{\omega}(s)ds
		\\
		\label{S3E5}
		-\frac{1}{\omega}\int_0^th^2_{\omega}(s)(b(x_{\omega}(s))\times \nabla p(x_{\omega}(s)))^{\operatorname{Tr}}\cdot \frac{\nabla_bb(x_{\omega}(s))}{|B(x_{\omega}(s))|}ds+\frac{t\mathcal{O}(1)}{\omega^2}.
	\end{gather}
	To handle the last two terms we observe first that $\nabla_bb=\operatorname{curl}(b)\times b+\nabla \frac{|b|^2}{2}=\operatorname{curl}(b)\times b$ and hence $(b\times \nabla p)\cdot (\operatorname{curl}(b)\times b)=(b\cdot \operatorname{curl}(b))(b\cdot \nabla p)-|b|^2(\nabla p \cdot \operatorname{curl}(b))=-\nabla p\cdot \operatorname{curl}(b)$ where we used that $b$ is a unit field and that $b$ and $\nabla p$ are orthogonal. Further we find $\operatorname{curl}(b)=\operatorname{curl}\left(\frac{B}{|B|}\right)=\frac{\operatorname{curl}(B)}{|B|}+\nabla |B|^{-1}\times B$ and so $-\nabla p\cdot \operatorname{curl}(b)=(\nabla p\times B)\cdot \nabla |B|^{-1}$ where we used that $\operatorname{curl}(B)$ and $\nabla p$ are also everywhere orthogonal.
	
	We can insert this in the last integral term in (\ref{S3E5}) to arrive at
	\begin{gather}
		\nonumber
			p(x_{\omega}(t))-p_0=\frac{\nabla p(x_{\omega}(t))\cdot (b(x_{\omega}(t))\times v_{\omega}(t))}{\omega|B_{x_{\omega}(t)}|}-\frac{\nabla p(x_0)\cdot (b(x_0)\times v_0)}{\omega|B(x_0)|}
		\\
		\nonumber
		-\frac{1}{\omega}\int_0^tv_{\omega}(s)\cdot (\nabla p(x_{\omega}(s))\times b(x_{\omega}(s)))(\nabla |B|^{-1})(x_{\omega}(s))\cdot v_{\omega}(s)ds
		\\
		\label{S3E6}
		-\frac{1}{\omega}\int_0^th^2_{\omega}(s)(\nabla p_{\omega}(s)\times b_{\omega}(s))\cdot \nabla |B|^{-1}(x_{\omega}(s))ds+\frac{t\mathcal{O}(1)}{\omega^2}.
	\end{gather}
	Now we rearrange the terms in the second line in (\ref{S3E6}) as follows
	\begin{gather}
		\nonumber
		\int_0^tv_{\omega}(s)\cdot (\nabla p(x_{\omega}(s))\times b(x_{\omega}(s)))(\nabla |B|^{-1})(x_{\omega}(s))\cdot v_{\omega}(s)ds
		\\
		\nonumber
		=-\int_0^t(v_{\omega}(s)\times b_{\omega}(s))\cdot \nabla p(x_{\omega}(s))(\nabla |B|^{-1})(x_{\omega}(s))\cdot v_{\omega}(s)ds.
	\end{gather}
	Now we can use the defining equations (\ref{E1}) and write $v_{\omega}\times b_{\omega}=\frac{\dot{v}_{\omega}}{\omega|B_{\omega}|}$ so that an integration by parts yields
	\begin{gather}
		\nonumber
		\int_0^tv_{\omega}(s)\cdot (\nabla p(x_{\omega}(s))\times b(x_{\omega}(s)))(\nabla |B|^{-1})(x_{\omega}(s))\cdot v_{\omega}(s)ds
		\\
		\label{S3E7}
		=\frac{t\mathcal{O}(1)}{\omega}+\int_0^t(v_{\omega}(s)\cdot \nabla p_{\omega}(s))\nabla |B|^{-1}(x_{\omega}(s))\cdot (v_{\omega}(s)\times b_{\omega}(s))ds.
	\end{gather}
	Now we use the relation (\ref{S3E1}) which states $B\times \operatorname{curl}(B)=\nabla p$ so that $\nabla p\times b=(B\times \operatorname{curl}(B))\times b=|B|\operatorname{curl}(B)-(b\cdot \operatorname{curl}(B))B$ and hence we obtain the following identities
	\begin{gather}
		\nonumber
		\int_0^t(v_{\omega}(s)\cdot \nabla p_{\omega}(s))\nabla |B|^{-1}(x_{\omega}(s))\cdot (v_{\omega}(s)\times b_{\omega}(s))ds
		\\
		\nonumber
		=\int_0^t(v_{\omega}(s)\times b_{\omega}(s))\cdot \operatorname{curl}(B)(x_{\omega}(s))(v_{\omega}(s)\times b_{\omega}(s))\cdot (\nabla |B_{\omega}|^{-1}) |B_{\omega}(s)|ds\text{ and }
		\\
		\nonumber
		\int_0^tv_{\omega}(s)\cdot (\nabla p(x_{\omega}(s))\times b(x_{\omega}(s)))(\nabla |B|^{-1})(x_{\omega}(s))\cdot v_{\omega}(s)ds
		\\
		\nonumber
		=\int_0^t(v_{\omega}(s)\cdot \operatorname{curl}(B_{\omega}))v_{\omega}(s)\cdot (\nabla |B_{\omega}|^{-1})|B_{\omega}(s)|ds-\int_0^t(v_{\omega}(s)\cdot B_{\omega}(s))(b_{\omega}(s)\cdot \operatorname{curl}(B_{\omega}))(v_{\omega}(s)\cdot \nabla |B_{\omega}|^{-1})ds.
	\end{gather}
	We can combine these two identities together with (\ref{S3E7}) to arrive at
	\begin{gather}
		\nonumber
		2\int_0^tv_{\omega}(s)\cdot (\nabla p(x_{\omega}(s))\times b(x_{\omega}(s)))(\nabla |B|^{-1})(x_{\omega}(s))\cdot v_{\omega}(s)ds
		\\
		\nonumber
		=\frac{t\mathcal{O}(1)}{\omega}-\int_0^th_{\omega}(s)(b_{\omega}(s)\cdot  \operatorname{curl}(B_{\omega}))(v_{\omega}(s)\cdot \nabla |B_{\omega}|^{-1})|B_{\omega}(s)|ds
		\\
		\nonumber
		+\int_0^t\left[((v_{\omega}\times b_{\omega})\cdot \operatorname{curl}(B_{\omega}))(v_{\omega}\times b_{\omega})+(v_{\omega}\cdot \operatorname{curl}(B_{\omega}))v_{\omega}\right]\cdot (\nabla |B_{\omega}|^{-1})|B_{\omega}|ds
		\\
		\nonumber
		=\frac{t\mathcal{O}(1)}{\omega}+\int_0^t\left[((v_{\omega}\times b_{\omega})\cdot \operatorname{curl}(B_{\omega}))(v_{\omega}\times b_{\omega})+(v^{\perp}_{\omega}\cdot \operatorname{curl}(B_{\omega}))v_{\omega}\right]\cdot (\nabla |B_{\omega}|^{-1})|B_{\omega}|ds.
	\end{gather}
	Now we observe that $|v_{\omega}(s)|=|v_0|$ so that $|v_{\omega}\times b_{\omega}|^2=|v_0|^2-h^2_{\omega}$ and $|v^{\perp}_{\omega}|^2=|v_{\omega}|^2-|v^{\parallel}_{\omega}|^2=|v_0|^2-h^2_{\omega}$. From this it follows that
	\begin{gather}
		\nonumber
		((v_{\omega}\times b_{\omega})\cdot \operatorname{curl}(B_{\omega}))(v_{\omega}\times b_{\omega})+(v^{\perp}_{\omega}\cdot \operatorname{curl}(B_{\omega}))v^{\perp}_{\omega}
		\\
		\nonumber
		=(|v_0|^2-h^2_{\omega}(s))(\operatorname{curl}(B)(x_{\omega}(s))-(b_{\omega}(s)\cdot \operatorname{curl}(B_{\omega}(s))b_{\omega})
	\end{gather}
	since $\left\{b_{\omega},\frac{v_{\omega}\times b_{\omega}}{\sqrt{|v_0|^2-h^2_{\omega}}},\frac{v^{\perp}_{\omega}}{\sqrt{|v_0|^2-h^2_{\omega}}}\right\}$ forms an orthonormal basis. Combining these considerations we find
	\begin{gather}
		\nonumber
		2\int_0^tv_{\omega}(s)\cdot (\nabla p(x_{\omega}(s))\times b(x_{\omega}(s)))(\nabla |B|^{-1})(x_{\omega}(s))\cdot v_{\omega}(s)ds
		\\
		\nonumber
		=\frac{t\mathcal{O}(1)}{\omega}+\int_0^t(|v_0|^2-h^2_{\omega}(s))\left(\operatorname{curl}(B_{\omega})-(b_{\omega}\cdot \operatorname{curl}(B_{\omega}))b_{\omega}\right)\cdot (\nabla |B_{\omega}|^{-1})|B_{\omega}(s)|ds
		\\
		\nonumber
		+\int_0^t(v^{\perp}_{\omega}\cdot \operatorname{curl}(B_{\omega}))v^{\parallel}_{\omega}(s)\cdot (\nabla |B_{\omega}|^{-1})|B_{\omega}|ds.
	\end{gather}
	We see that the last integral is of order $\frac{t\mathcal{O}(1)}{\omega}$ by means of our usual integration by parts argument. We finally observe that $|B|(\operatorname{curl}(B)-(b\cdot \operatorname{curl}(B))b)=b\times (\operatorname{curl}(B)\times B)=\nabla p\times b$ due to the identity $B\times \operatorname{curl}(B)=\nabla p$. We therefore conclude
	\begin{gather}
		\nonumber
			2\int_0^tv_{\omega}(s)\cdot (\nabla p(x_{\omega}(s))\times b(x_{\omega}(s)))(\nabla |B|^{-1})(x_{\omega}(s))\cdot v_{\omega}(s)ds
		\\
		\nonumber
		=\frac{t\mathcal{O}(1)}{\omega}+\int_0^t(|v_0|^2-h^2_{\omega}(s))(\nabla p_{\omega}(s)\times b_{\omega}(s))\cdot \nabla |B|^{-1}(x_{\omega}(s))ds.
	\end{gather}
	We can insert this into (\ref{S3E6}) and find
	\begin{gather}
		\nonumber
		p(x_{\omega}(t))-p_0=\frac{\nabla p(x_{\omega}(t))\cdot (b(x_{\omega}(t))\times v_{\omega}(t))}{\omega|B(x_{\omega}(t))|}-\frac{\nabla p(x_0)\cdot (b(x_0)\times v_0)}{\omega|B(x_0)|}
		\\
		\label{S3E8}
		-\frac{1}{\omega}\int_0^t\frac{|v_0|^2+h^2_{\omega}(s)}{2}(\nabla p(x_{\omega}(s))\times b(x_{\omega}(s)))\cdot \nabla |B|^{-1}(x_{\omega}(s))ds+\frac{t\mathcal{O}(1)}{\omega^2}
	\end{gather}
	where the constant in $\mathcal{O}(1)$ will depend on $\Omega,B$ but is independent of $\omega$ or $t$ (as long as the trajectory $x_{\omega}$ stays within $\Omega$ up to time $t$).
	
	It now follows from \Cref{IT1} and its proof that $x_{\omega}$ as well as $h_{\omega}$ converge in $C^{0,\alpha}_{\operatorname{loc}}(\mathbb{R})$ to the limit trajectory $x(t)$ and limit function $h(t)$, where $\dot{x}(t)=h(t)b(x(t))$ and $x(0)=x_0$. Knowing this, we conclude from (\ref{S3E8})
	\begin{gather}
		\nonumber
		p(x_{\omega}(t))-p_0=\frac{\nabla p(x(t))\cdot (b(x(t))\times v_{\omega}(t))}{\omega|B(x(t))|}-\frac{\nabla p(x_0)\cdot (b(x_0)\times v_0}{\omega|B(x_0)|}
		\\
		\nonumber
		-\frac{1}{\omega}\int_0^t\frac{|v_0|^2+h^2(s)}{2}(\nabla p(x(s))\times b(x(s)))\cdot \nabla (|B|^{-1})(x(s))ds
		\\
		\label{S3E9}
		+\frac{t\mathcal{O}(1)}{\omega^2}+\frac{\mathcal{O}(1)}{\omega}\left(|x_{\omega}(t)-x(t)|+|h_{\omega}(t)-h(t)|\right).
	\end{gather}
	Lastly, using the definition of the magnetic moment $\mu(t)=\frac{|v_0|^2-h^2(t)}{2|B(x(t))|}$ and that it is preserved in time, c.f. \Cref{IC2}, we arrive at the expression in \Cref{IT3}. The convergence rate claim then follows from \Cref{IT1} and its proof, c.f. \ref{S2E18},(\ref{S2E19}) where we also established a convergence rate for $|h_{\omega}(t)-h(t)|$. Note that we may take $\gamma=0$ because we only consider times for which $x_{\omega}(t)$ remains in $\Omega$ so that the terms of the form $|B(x_{\omega}(t))|^{r}$ for $r\in \mathbb{R}$ may be estimated independently of $\omega$ and $t$ and hence there is also no need to impose decay conditions on the magnetic field in question.
\end{proof}
\begin{proof}[Proof of \Cref{S3L2}]
	The derivation is similar in spirit to that of \Cref{S2L2}. We start by exploiting the identity
	\begin{gather}
		\nonumber
		0=\nabla (b\cdot (v\times b))=\nabla_{v\times b}b+\nabla_b(v\times b)+(v\times b)\times \operatorname{curl}(b)+b\times (\operatorname{curl}(v\times b))
		\\
		\label{S3E10}
		=\nabla_{v\times b}b+v\times \nabla_bb+(v\times b)\times \operatorname{curl}(b)+b\times \operatorname{curl}(v\times b)
	\end{gather}
	where we used that $\nabla_b(v\times b)=\nabla_bv\times b+v\times \nabla_bb$ and $\nabla_bv=0$ since $v$ is a fixed vector independent of $x$. We now use the identity $\operatorname{curl}(v\times b)=-[v,b]+v\operatorname{div}(b)-b\operatorname{div}(v)=-\nabla_vb+v\operatorname{div}(b)$ where we used again that $v$ is independent of $x$. We conclude
	\begin{gather}
		\nonumber
		\nabla_{v\times b}b=\nabla_bb\times v+\operatorname{curl}(b)\times (v\times b)-\nabla_vb\times b+\operatorname{div}(b)v\times b.
	\end{gather}
	Now we note that $\nabla_bb=\operatorname{curl}(b)\times b$ since $b$ is a unit field and so we find $\nabla_bb\times v+\operatorname{curl}(b)\times (v\times b)=(\operatorname{curl}(b)\times v)\times b$. We obtain
	\begin{gather}
		\label{S3E11}
		\nabla_{v\times b}b=(\operatorname{curl}(b)\times v)\times b-\nabla_vb\times b+\operatorname{div}(b)v\times b.
	\end{gather}
	Now we observe that $Db\cdot (v\times b)=\nabla_{v\times b}b$ so that
	\begin{gather}
		\label{S3E12}
		(Db\cdot (v\times b))\times (v\times b)=((\operatorname{curl}(b)\times v)\times b)\times (v\times b)-(\nabla_vb\times b)\times (v\times b).
	\end{gather}
	We compute $(\operatorname{curl}(b)\times v)\times b=(b\cdot \operatorname{curl}(b))v-(b\cdot v)\operatorname{curl}(b)$ and so $((\operatorname{curl}(b)\times v)\times b)\times (v\times b)=(b\cdot v)(v\cdot \operatorname{curl}(b))b-|v|^2(b\cdot \operatorname{curl}(b))b$. On the other hand $(\nabla_vb\times b)\times (v\times b)=(\nabla_vb\cdot (v\times b))b$ since $b$ and $v\times b$ are orthogonal. Therefore (\ref{S3E12}) becomes
	\begin{gather}
		\nonumber
		(Db\cdot (v\times b))\times (v\times b)=(b\cdot v)(v\cdot \operatorname{curl}(b))b-|v|^2(b\cdot \operatorname{curl}(b))b-(\nabla_vb\cdot (v\times b))b
	\end{gather}
	which is the claimed identity.
\end{proof}
\section{Proof of \Cref{IT6}}
\begin{proof}[Proof of \Cref{IT6}]
	According to \Cref{IT3} we have to explicitly compute the integral
	\begin{gather}
		\nonumber
		\int_0^t\left(\frac{|v_0|^2}{|B(x(s))|}-\mu_0\right)(B(x(s))\times \nabla p(x(s)))\cdot \nabla |B|^{-1}(x(s))ds.
	\end{gather}
	To this end we recall from the proof of \Cref{IT1} that
	\begin{gather}
		\nonumber
		\dot{x}(t)=h(t)b(x(t))\text{ and }h^2(t)=|v_0|^2-2\mu_0|B(x(t))|
	\end{gather}
	which follows from the fact that the magnetic moment is preserved, c.f. \Cref{IC2}, and where $b(y)=\frac{B(y)}{|B(y)|}$. We recall that we assume the particle to be passing, meaning that $(x_0,v_0)$ is such that $|v_0|^2-2\mu_0|B(y)|>0$ for every $y\in \Sigma:=\{z\in \Omega\mid p(z)=p(x_0)\}$. Since $B\parallel \Sigma$ and $x(t)$ follows magnetic field lines, we see that $|v_0|^2-2\mu_0|B(x(t))|\geq c_0>0$ for all $t$. We hence obtain
	\begin{gather}
		\nonumber
		\int_0^t\left(\frac{|v_0|^2}{|B(x(s))|}-\mu_0\right)(B(x(s))\times \nabla p(x(s)))\cdot \nabla |B|^{-1}(x(s))ds
		\\
		\nonumber
		=\int_0^t\frac{|v_0|^2-|B(x(s))|\mu_0}{\sqrt{|v_0|^2-2\mu_0|B(x(s))|}}\dot{x}(s)\cdot \left(\nabla p(x(s))\times \nabla |B|^{-1}(x(s))\right)ds
		\\
		\nonumber
		=\int_{x[0,t]}\frac{|v_0|^2-\mu_0|B|}{\sqrt{|v_0|^2-2\mu_0|B|}}\nabla p\times \nabla |B|^{-1}
		\end{gather}
	where the latter is a line integral and we assume w.l.o.g. that $h(0)>0$. Since $x$ follows a magnetic field line and so do the integral curves of $\eta=\frac{B}{|B|^2}$ we see that for every $t>0$ there will be a unique $\rho(t)>0$ such that if we let $\gamma$ denote the integral curve of $\eta$ starting at $x_0$, then $\gamma([0,\rho])=x([0,t])$ and both curves (in case the field lines are periodic) traversed the point $x_0$ the same amount of times. We can then write
	\begin{gather}
	\nonumber
	\int_0^t\left(\frac{|v_0|^2}{|B(x(s))|}-\mu_0\right)(B(x(s))\times \nabla p(x(s)))\cdot \nabla |B|^{-1}(x(s))ds
	\\
	\nonumber
	=\int_0^{\rho}\frac{|v_0|^2-\mu_0|B(\gamma(s))|}{\sqrt{|v_0|^2-2\mu_0|B(\gamma(s))|}}\xi(\gamma(s))\cdot \nabla |B|^{-1}(\gamma(s))ds
\end{gather}
where $\xi:=\eta \times \nabla p$. We can now define the function
\begin{gather}
	\nonumber
	F(\sigma):=\int_{|B(x_0)|}^{\sigma}\frac{\mu_0\tau-|v_0|^2}{\sqrt{|v_0|^2-2\mu_0\tau}}\frac{1}{\tau^2}d\tau
\end{gather}
and observe that
\begin{gather}
	\nonumber
	\int_0^t\left(\frac{|v_0|^2}{|B(x(s))|}-\mu_0\right)(B(x(s))\times \nabla p(x(s)))\cdot \nabla |B|^{-1}(x(s))ds
	\\
	\nonumber
	=\int_0^\rho \xi(\gamma(s))\cdot \nabla (F(|B|)(\gamma(s))ds.
\end{gather}
We now exploit the fact that $B$ is quasi-symmetric, meaning that if we change into Boozer-coordinates $(p,\phi,\theta)$ we find $\xi(p,\phi,\theta)=a(p)\partial_{\phi}+c(p)\partial_{\theta}$ and $|B(p,\phi,\theta)|=f(p,M\phi+N\theta)$ for a suitable function $f:I\times S^1\rightarrow [0,\infty)$. Since we work on a fixed pressure level set, we will from now on suppress the $p$ dependence in the notation. We can then set $\lambda(s):=(F\circ f)(s)$ and see that $\xi(p,\phi,\theta)\cdot \nabla (F(|B|))(p,\phi,\theta)=a\partial_{\phi}(\lambda(M\phi+N\theta))+b\partial_{\theta}(\lambda(M\phi+N\theta))=(aM+cN)\lambda^{\prime}(M\phi+N\theta)$. We also observe that in Boozer coordinates the field line $\gamma$ of $\eta$ has the expression $\gamma(s)=(\alpha s+\phi_0,\beta s+\theta_0)$ where the point $(\phi_0,\theta_0)$ corresponds to the initial condition $x_0$. We let $m:=\alpha M+\beta N$, $c_0:=M\phi_0+N\theta_0$ and overall arrive at
\begin{gather}
	\label{S4E1}
	\int_0^t\left(\frac{|v_0|^2}{|B(x(s))|}-\mu_0\right)(B(x(s))\times \nabla p(x(s)))\cdot \nabla |B|^{-1}(x(s))ds=(aM+cN)\int_0^{\rho}\lambda^{\prime}(ms+c_0)ds.
\end{gather}
Now, if $m\neq 0$ we can perform a change of coordinates leading us to
\begin{gather}
	\nonumber
	\int_0^{\rho}\lambda^{\prime}(ms+c_0)ds=\frac{\int_{c_0}^{m\rho+c_0}\lambda^{\prime}(\tau)d\tau}{m}=\frac{\lambda(m\rho+c_0)-\lambda(c_0)}{m}=\frac{F(|B|(\gamma(\rho)))}{m}=\frac{F(|B|(x(t)))}{m}
\end{gather}
where we used that $F(|B(x_0)|)=0$ and that $\gamma(\rho)=x(t)$ by choice of $\rho$. We thus arrive at
\begin{gather}
	\nonumber
	\int_0^t\left(\frac{|v_0|^2}{|B(x(s))|}-\mu_0\right)(B(x(s))\times \nabla p(x(s)))\cdot \nabla |B|^{-1}(x(s))ds=\frac{aM+cN}{\alpha M+\beta N}F(|B|(x(t))).
\end{gather}
Lastly we observe that we can find an explicit primitive expressing
\begin{gather}
	\nonumber
	F(s)=\frac{\sqrt{|v_0|^2-2\mu_0s}}{s}-\frac{\sqrt{|v_0|^2-2\mu_0|B(x_0)|}}{|B(x_0)|}.
\end{gather}
Recall, keeping in mind magnetic moment preservation, c.f. \Cref{IC2}, that $\sqrt{|v_0|^2-2\mu_0|B(x(t))|}=h(t)$ and hence the claimed formula follows in the case $m\neq 0$.
\newline
Now on the other hand, if $m=0$, then (\ref{S4E1}) becomes
\begin{gather}
	\nonumber
	\int_0^t\left(\frac{|v_0|^2}{|B(x(s))|}-\mu_0\right)(B(x(s))\times \nabla p(x(s)))\cdot \nabla |B|^{-1}(x(s))ds=(aM+cN)\rho(t)\lambda^{\prime}(c_0).
\end{gather}
We recall that $\lambda(s)=(F\circ f)(s)$ and so $\lambda^{\prime}(c_0)=F^{\prime}(|B(x_0)|)f^{\prime}(\phi_0 M+\theta_0 N)$ which also immediately yields the claimed formula.
\newline
We are left with establishing a more explicit relation between $\rho(t)$ and $t$. To this end we compute
\begin{gather}
	\nonumber
	\rho=\int_{\gamma[0,\rho]}B=\int_{x[0,t]}B=\int_0^t\dot{x}(s)\cdot B(x(s))ds=\int_0^t h(s)|B(x(s))|
\end{gather}
where we used the defining properties of $\rho$. Using the magnetic moment we can express $h$ in terms of $|B(x(s))|$ which concludes the proof.
\end{proof}
\section*{Acknowledgements}
This work has been partly supported by the  ANR-DFG project “CoRoMo”
ANR-22-CE92-0077-01, and has received financial support from the CNRS through the MITI interdisciplinary programs. This work has been supported by the Inria AEX StellaCage. The research was supported in part by the MIUR Excellence Department Project awarded to Dipartimento di Matematica, Università di Genova, CUP D33C23001110001.
\bibliographystyle{plain}
\bibliography{mybibfileNOHYPERLINK}
\footnotesize
\end{document}